\newif\iffull\fulltrue

\documentclass[acmsmall,nonacm]{acmart}

\iffull
\renewcommand\footnotetextcopyrightpermission[1]{} 
\else
\fi

\AtBeginDocument{%
  }

\usepackage{xspace}
\usepackage{algorithmicx}
\usepackage{subcaption,graphicx,url,multirow,multicol,mathtools}
\usepackage[noend]{algpseudocode}
\usepackage{balance, afterpage}
\usepackage{paralist}
\usepackage{booktabs}
\usepackage{color}
\usepackage{arydshln}
\usepackage{enumitem}
\usepackage[most]{tcolorbox}

\usepackage{stmaryrd,amsthm}
\usepackage{bm}

\newcommand{\myparagraph}[1]{\vspace{2pt}\noindent {\bf #1.}}

\mathchardef\mhyphen="2D

\newcommand{\whp}[1]{w.h.p.}
\newcommand{\defn}[1]{{\bf{\emph{#1}}}}

\usepackage[ruled]{algorithm}

\newcommand{\BigO}[1]{\ensuremath{O(#1)}}
\newcommand{\BigOmega}[1]{\ensuremath{\Omega(#1)}}

\newcommand{\slfrac}[2]{\left.#1\middle/#2\right.}

\makeatletter
\newcommand{\newreptheorem}[2]{\newtheorem*{rep@#1}{\rep@title}\newenvironment{rep#1}[1]{\def\rep@title{#2 \ref*{##1}}\begin{rep@#1}}{\end{rep@#1}}}
\makeatother

\newreptheorem{theorem}{Theorem}
\newreptheorem{lemma}{Lemma}

\algblock{ParFor}{EndParFor}
\algnewcommand\algorithmicparfor{\textbf{parfor}}
\algnewcommand\algorithmicpardo{\textbf{do}}
\algnewcommand\algorithmicendparfor{}
\algrenewtext{ParFor}[1]{\algorithmicparfor\ #1\ \algorithmicpardo}
\algrenewtext{EndParFor}{\algorithmicendparfor}
\algtext*{EndParFor}

\algdef{SE}[DOWHILE]{Do}{doWhile}{\algorithmicdo}[1]{\algorithmicwhile\ #1}

\algdef{SE}[SUBALG]{Indent}{EndIndent}{}{\algorithmicend\ }%
\algtext*{Indent}
\algtext*{EndIndent}

\definecolor{ao}{rgb}{0.0, 0.5, 0.0}

\newcommand{\algname}[1]{\textnormal{\textsc{#1}}}

\newcommand{\ournd}{\algname{arb-nucleus}\xspace}
\newcommand{\ourcountrec}{\algname{rec-list-cliques}\xspace}
\newcommand{\sariyuce}{Sariy{\"{u}}ce\xspace}
 
\newcommand{\cnucleus}[3]{$#1$-$(#2, #3)$ nucleus} 
\newcommand{\nucleusdecomp}[2]{$(#1, #2)$ nucleus decomposition}

\newcommand{\ourndh}{\algname{arb-nucleus-hierarchy}\xspace}
\newcommand{\ourapprox}{\algname{arb-approx-nucleus-hierarchy}\xspace}
\newcommand{\ourapproxnd}{\algname{approx-arb-nucleus}\xspace} 
\newcommand{\ourlink}{\algname{link}\xspace}
\newcommand{\ourilink}{\algname{link-basic}\xspace}
\newcommand{\ourelink}{\algname{link-efficient}\xspace}
\newcommand{\ourtree}{\algname{construct-tree}\xspace}
\newcommand{\ouritree}{\algname{construct-tree-basic}\xspace}
\newcommand{\ouretree}{\algname{construct-tree-efficient}\xspace}
\newcommand{\cas}{\algname{compare-and-swap}\xspace}
\newcommand{\ourframework}{\algname{arb-nucleus-decomp-hierarchy-framework}\xspace}

\newcommand{\ourte}{\algname{anh-te}\xspace}
\newcommand{\ourel}{\algname{anh-el}\xspace}
\newcommand{\ourbl}{\algname{anh-bl}\xspace}

\newcommand{\ourapproxte}{\algname{approx-anh-te}\xspace}
\newcommand{\ourapproxel}{\algname{approx-anh-el}\xspace}
\newcommand{\ourapproxbl}{\algname{approx-anh-bl}\xspace}

\definecolor{change}{rgb}{0,0,0}

\usepackage{footnote}
\makesavenoteenv{algorithm}

\makeatletter
\renewenvironment{proof}[1][\proofname]{\par
  \vspace{-0.5\topsep}%
  \pushQED{\qed}%
  \normalfont
  \topsep0pt \partopsep0pt %
  \trivlist
  \item[\hskip\labelsep
        \scshape
    #1\@addpunct{.}]\ignorespaces
}{%
  \popQED\endtrivlist\@endpefalse
  \addvspace{4pt} %
}
\makeatother

\usepackage{xcolor}
\definecolor{lightred}{RGB}{252, 172, 167}
\definecolor{lightblue}{RGB}{207, 233, 252}
\definecolor{lightgreen}{RGB}{176, 252, 167}
\makeatletter
\newcommand{\algcolor}[2]{%
  \hskip-\ALG@thistlm\colorbox{#1}{\parbox{\dimexpr\linewidth-4\fboxsep}{\hskip\ALG@thistlm\relax #2}}%
}

\makeatother

\usepackage{soul}
\sethlcolor{lightblue}

\newcommand{\algblue}[1]{\colorbox{lightblue}{#1}}

\begin{document}
\title[
Parallel Algorithms for Hierarchical Nucleus Decomposition] 
{Parallel Algorithms for Hierarchical Nucleus Decomposition}

\author{Jessica Shi}
\affiliation{\institution{MIT CSAIL} \country{USA}}
\email{jeshi@mit.edu}
\author{Laxman Dhulipala}
\affiliation{\institution{University of Maryland, College Park} \country{USA}}
\email{laxman@umd.edu}
\author{Julian Shun}
\affiliation{\institution{MIT CSAIL} \country{USA}}
\email{jshun@mit.edu}

\begin{abstract}
Nucleus decompositions have been shown to be a useful tool for finding dense subgraphs. The coreness value of a clique represents its density based on the number of other cliques it is adjacent to.
One useful output of nucleus decomposition is to generate a hierarchy among dense subgraphs at different resolutions. 
However, existing parallel algorithms for nucleus decomposition
do not generate this hierarchy, and only compute the coreness values. This paper presents a scalable parallel algorithm for hierarchy construction, with practical optimizations, 
such as interleaving the coreness computation with hierarchy construction and using a concurrent union-find data structure in an innovative way to generate the hierarchy.
We also introduce a parallel approximation algorithm for nucleus decomposition, which achieves much lower span in theory and better performance in practice. 
We prove strong theoretical bounds on the work and span (parallel time) of our algorithms.

On a 30-core machine with two-way hyper-threading,
our parallel hierarchy construction algorithm achieves up to a 58.84x speedup over the 
state-of-the-art sequential hierarchy construction algorithm by \sariyuce \text{et al.} and up to a 30.96x self-relative parallel speedup. 
On the same machine, our approximation algorithm achieves a 3.3x speedup over our exact algorithm, while generating coreness estimates with a multiplicative error of 1.33x on average.

\end{abstract}

\maketitle

\iffull
\pagestyle{plain} 
\else
\fi

\section{Introduction}
Dense subgraph and substructure detection is a fundamental tool in graph mining,
with many applications in areas including social network analysis~\cite{Ye2011,Venica2021}, fraud detection~\cite{gibson2005discovering}, and computational biology~\cite{bader2003automated,fratkin2006motifcut}.
\sariyuce{} \textit{et al.}~\cite{Sariyuce2017} introduced the 
{\em nucleus decomposition problem},  a generalization of the $k$-core~\cite{seidman83network, matula83smallest} and $k$-truss~\cite{cohen2008trusses} problems which better captures higher-order structures in graphs.
In this problem, a 
$c$-$(r,s)$ nucleus is defined to be the maximal induced subgraph such that every
$r$-clique in the subgraph is contained in at least $c$ $s$-cliques.
The goal of the \nucleusdecomp{r}{s} problem is to (1) identify for each $r$-clique in the graph, the
largest $c$ such that it is in a $c$-$(r,s)$ nucleus (known as the coreness value) and 
(2) generate a {\em nucleus hierarchy} over the nuclei, where for $c' < c$
a $c'$-$(r,s)$ nucleus $A$ is a descendant of a 
$c$-$(r,s)$ nucleus $B$ if $A$ is a subgraph of $B$.
Figure~\ref{fig:graph} shows the hierarchy for $(1,3)$ nucleus.
The nucleus hierarchy is an unsupervised method for revealing dense substructures at different resolutions in the graph.
Since the hierarchy is a tree, it is easy to visualize and explore as part of structural graph analysis tasks~\cite{SaPi16}.

\sariyuce{} and Pinar present the first algorithm for solving the nucleus decomposition problem~\cite{SaPi16}. 
However, their algorithm is sequential, and in order to scale to the large graph sizes of today, it is important to design parallel algorithms that take advantage of modern parallel hardware. 
The existing parallel algorithms for nucleus decomposition include those by \sariyuce{} \textit{et al.}~\cite{sariyuce2017parallel} and by Shi \textit{et al.}~\cite{shi2021nucleus}, but these algorithms only compute the coreness values. Importantly, they do not generate the hierarchy, which limits their applicability and which is non-trivial to generate in parallel.
In this paper, we design the first work-efficient parallel algorithm for hierarchy construction in nucleus decomposition,
where the work, or the total number of operations, matches the best sequential time complexity.
Our algorithm runs in $O(m\alpha^{s-2})$ expected work and $O(k \log n + \rho_{(r,s)}(G)\log n + \log^2 n)$ span \whp{} (parallel time), 
where $\alpha$ is the arboricity of the input graph $G$, $\rho_{(r,s)}(G)$ is the $(r,s)$ peeling complexity of $G$, and $k$ is the maximum $(r, s)$-clique core number in $G$.
The key to our theoretical efficiency is our careful construction of subgraphs representing the $s$-clique-connectivity of $r$-cliques, that allows us to exploit linear-work graph connectivity instead of more expensive union-finds as used in prior work. 
Our approach gives as a byproduct the most theoretically-efficient serial algorithm for computing the hierarchy, improving upon the previous best known serial bounds by \sariyuce{} and Pinar~\cite{SaPi16}.

We also present a practical parallel algorithm that interleaves the hierarchy construction with the computation of the coreness values. 
Prior work by \sariyuce{} and Pinar~\cite{SaPi16} also includes a (serial) interleaved hierarchy algorithm. 
However, their algorithm requires storing 
all adjacent $r$-cliques with different coreness values, which could
potentially be proportional to the number of $s$-cliques in $G$ (i.e., use $O(m\alpha^{s-2})$ space), and which results in sequential dependencies in their post-processing step to construct the hierarchy. 
Our parallel algorithm fully interleaves the hierarchy construction with the coreness computation, and uses only two additional arrays of size proportional to the number of $r$-cliques in $G$. Our algorithm uses a concurrent union-find data structure in an innovative way.
Also, our post-processing step to construct the hierarchy tree is fully parallel. 
Our main insight is a technique to fully extract the connectivity information from adjacent $r$-cliques with different core numbers while computing the coreness values.

Note that the span of the above algorithms can be large for graphs with large peeling complexity ($\rho_{(r,s)}(G)$).  
We introduce an approximate algorithm for nucleus decomposition and show that it can achieve work efficiency and polylogarithmic span. Our algorithm relaxes the peeling order by allowing all $r$-cliques within a $(\binom{s}{r}+\epsilon)$ factor of the current value of $k$ to be peeled in parallel. 
We show that our algorithm generates coreness estimates that are an $(\binom{s}{r}+\epsilon)$-approximation of the true coreness values. 

We experimentally study our parallel algorithms on real-world graphs using different $(r,s)$ values, for up to $r < s \leq 7$. 
On a 30-core machine with two-way hyper-threading, our exact algorithm which generates both the coreness numbers and the hierarchy
achieves up to a 30.96x self-relative parallel speedup, and a 3.76--58.84x speedup over the state-of-the-art sequential algorithm by \sariyuce{} and Pinar~\cite{SaPi16}. 
In addition, we show that on the same machine our approximate algorithm is up to 3.3x faster than our exact algorithm for computing coreness values, while generating coreness estimates with a multiplicative error of 1--2.92x (with a median error of 1.33x).
Our algorithms are able to compute
the $(r,s)$ nucleus decomposition hierarchy for $r>3$ and $s>4$ on graphs with over a hundred million edges for the first time.

We summarize our contributions below:
\begin{itemize}[topsep=1pt,itemsep=0pt,parsep=0pt,leftmargin=10pt]
  \item The first exact and approximate parallel algorithms for nucleus decomposition that generates the hierarchy with strong theoretical bounds on the work, span, and approximation guarantees.
  \item A number of practical optimizations that lead to fast implementations of our algorithms.
  \item Experiments showing that our new exact algorithm achieves 3.7--58.8x speedups over state of the art on a 30-core machine with two-way hyper-threading.
  \item Experiments showing that our new approximate algorithm achieves up to 3.3x speedups over our exact algorithm, while generating coreness estimates with a multiplicative error of 1--2.9x (with a median error of 1.3x).
\end{itemize}

Our code is publicly available at \url{https://github.com/jeshi96/arb-nucleus-hierarchy}. 
\section{Related Work}\label{sec:related}
The $k$-core~\cite{seidman83network,matula83smallest} and $k$-truss problems~\cite{cohen2008trusses,zhang2012extracting,zhao2012large} are classic problems relating to dense substructures in graphs.
Many algorithms have been developed for $k$-cores and $k$-trusses in the static~\cite{Ghaffari2019,montresor2012distributed, farach2014computing, khaouid2015k, Kabir2017, DhBlSh17,Wen2019,DLRSSY22, Wang2012, chen2014distributed, Zou16, kabir2017parallel, smith2017truss, ChLaSuWaLu20,CoDeGrMaVe18, LoSpKuSrPoFrMc18,  BlLoKi19} and dynamic~\cite{Li2014, Wen2019,Zhang2017,sariyuce2016incremental,Akbas2017,Huang2014, aridhi2016distributed, lin2021dynamickcore,Luo2021,Hua2020,SCS20,Jin2018,Luo2020,Zhang2019a,Huang2015,LSYDS22} settings.
Similar ideas have been studied in bipartite graphs~\cite{Shi2020,Wang2020,lakhotia20receipt,SariyuceP18, liu2020alphabeta,Huang2023}.

A related problem is the $k$-clique densest subgraph problem~\cite{Tsourakakis15}, which defines the density of subgraphs based on the number of $k$-cliques in it, rather than the number of edges. Fang \textit{et al.}~\cite{FaYuChLaLi19} further generalize this notion to arbitrary fixed-sized subgraphs (although they only present experiments for cliques). Similar to $k$-core and $k$-truss, algorithms for approximating the $k$-clique densest subgraph are based on iteratively peeling (removing) elements from the graph.
Shi \textit{et al.}~\cite{shi2020parallel} present efficient parallel algorithms for solving the $k$-clique densest subgraph problem.

\sariyuce{} \textit{et al.}\ define the \nucleusdecomp{r}{s} problem and show that it can be used to find higher quality dense substructures in graphs that previous approaches. 
They provide efficient sequential~\cite{Sariyuce2017,SaPi16} and parallel algorithms~\cite{sariyuce2017parallel} for this problem. 
\sariyuce{} \textit{et al.}~\cite{sariyuce2017parallel} present two parallel algorithms for the nucleus decomposition problem: (1) a global peeling-based algorithm and (2) a local update model that iterates until convergence. 
They also introduce a sequential algorithm for constructing the nucleus decomposition hierarchy~\cite{SaPi16}, but as far as we know there are no parallel algorithms for constructing this hierarchy.
Their algorithm is not work-efficient as it uses union-find. We also note that their space-usage can in theory be as large as $O(n\alpha^{s-2})$, i.e., proportional to the number of $s$-cliques in the graph.
Chu \textit{et al.}~\cite{Chu2022seqkcore} present a parallel algorithm for generating the $k$-core decomposition hierarchy, which is a special case of nucleus decomposition for $r=1$ and $s=2$, but their algorithm does not trivially generalize to higher $r$ and $s$.
Their serial and parallel algorithms both use union-find and run in $O(m \alpha(n))$ work (where $\alpha(n)$ is the inverse Ackermann function), which is not work-efficient, and their parallel algorithm has depth that depends on the peeling-complexity of the input.
Shi \textit{et al.}~\cite{shi2021nucleus} present an improved parallel algorithm for nucleus decomposition, with good theoretical guarantees and practical performance, but it does not generate the hierarchy.
Their algorithm is work-efficient, in that it runs in work proportional to enumerating all $s$-cliques, i.e., $O(m\alpha^{s-2})$ work. They leverage a work-efficient parallel clique counting algorithm~\cite{shi2020parallel}, along with a multi-level hash table structure to store data associated with cliques space efficiently, and techniques for  traversing this structure in a cache-friendly manner.
More recently, \sariyuce{} 
generalizes 
the $r$-cliques and $s$-cliques in nucleus decomposition to any pair of subgraphs~\cite{sariyuce2021motif}.
There has also been work on nucleus decomposition in probabilistic graphs~\cite{esfahani2020nucleus}.

\begin{figure}[t]
   \centering
   \vspace{-1em}
   \includegraphics[width=0.4\textwidth]{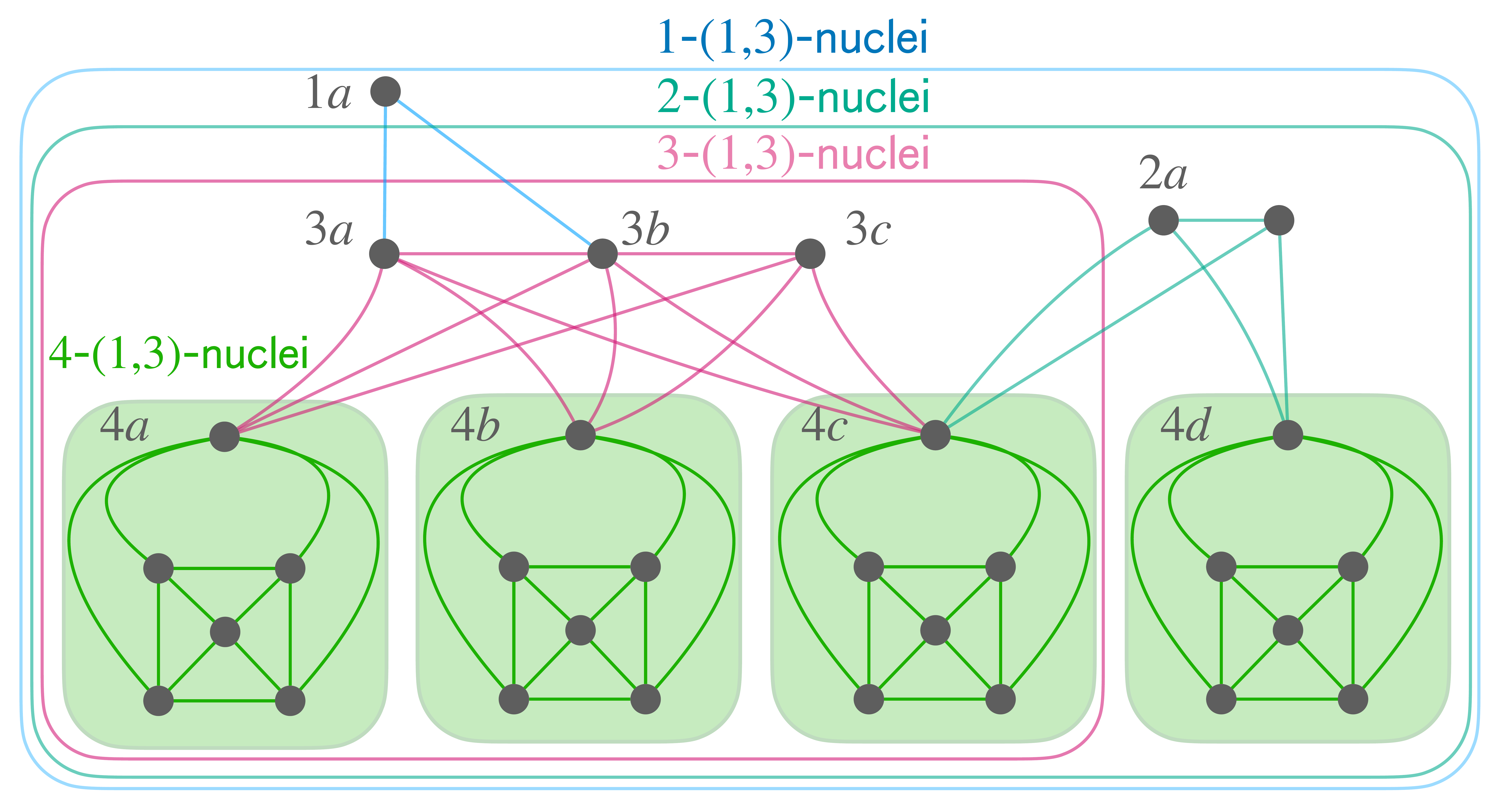}
   \vspace{-1em}
   \caption{
An example of the $(1, 3)$ nucleus hierarchy.
   }
    \label{fig:graph}
\end{figure}

\section{Preliminaries}\label{sec:prelims}
\myparagraph{Model of Computation} 
We use the work-span model of computation for our theoretical analysis, which is a standard model for shared-memory parallel algorithms~\cite{JaJa92,CLRS}.
The \defn{work} $\mathcal{W}$ of an algorithm is the total number of operations, and the \defn{span} (parallel time) $\mathcal{S}$ of an algorithm is the longest dependency path. 
Using a randomized work-stealing scheduler~\cite{BlLe99} on $\mathcal{P}$ processors, we can obtain a running time of $\mathcal{W}/\mathcal{P} + O(\mathcal{S})$ in expectation.
Our goal is to develop \defn{work-efficient} parallel algorithms under this model, or algorithms with a work complexity that asymptotically matches the best-known sequential time complexity for the given problem.

\myparagraph{Graph Notation and Definitions} We consider simple and undirected graphs $G =(V,E)$. We let $n = |V|$ and $m = |E|$. For analysis, we assume $m=\BigOmega{n}$.
The \defn{arboricity ($\bm{\alpha}$)} of a graph is the minimum number
of spanning forests needed to cover the graph. $\alpha$ is
upper bounded by $\BigO{\sqrt{m}}$ and lower bounded by
$\BigOmega{1}$~\cite{ChNi85}.
An \defn{$k$-clique} is a set of vertices such that all ${k \choose 2}$ edges exist among them. Enumerating $k$-cliques can be done in $\BigO{m\alpha^{k-2}}$ work and $\BigO{\log^2 n}$ span with high probability (\whp{})~\cite{shi2020parallel}.\footnote{We
  say $O(f(n))$ \defn{with high probability (\whp{})} to indicate
  $O(cf(n))$ with probability at least $1-n^{-c}$ for $c \geq 1$,
  where $n$ is the input size.}  

A \defn{\cnucleus{c}{r}{s}} is a maximal subgraph $H$ of an undirected graph formed by the union of $s$-cliques $S$, such that each $r$-clique $R$ in $H$ has induced $s$-clique degree at least $c$ (i.e., each $r$-clique is contained within at least $c$ induced $s$-cliques). 
The goal of the \defn{\nucleusdecomp{r}{s} problem} 
is to compute the following: (1) the \defn{$(r,s)$-clique core number} of each $r$-clique $R$, or the maximum $c$ such that $R$ is contained within a \cnucleus{c}{r}{s} and (2) a hierarchy over the nuclei, where for $c'<c$, a $c'$-$(r,s)$ nucleus $A$ is a descendant of a $c$-$(r,s)$ nucleus $B$ if $A$ is a subgraph of $B$. 
In this paper, similar to all prior work on nucleus computations, we take $r$ and $s$ to be constants and ignore constants depending on these values in our bounds.
The $k$-core and $k$-truss problems correspond to the $k$-$(1,2)$ and $k$-$(2,3)$ nucleus, respectively.
Figure~\ref{fig:graph} shows the hierarchy for $(1,3)$ nucleus.
We call two $r$-cliques \defn{$s$-clique-adjacent} if there exists an $s$-clique $S$ such that both $r$-cliques are subgraphs of $S$.
We also define the \defn{$s$-clique-degree} of a $r$-clique $R$ to be the number of $s$-cliques that $R$ is contained within.

We also consider approximate nucleus decompositions in this paper.
If the true $(r,s)$-clique core number of an $r$-clique $R$ is $k_R$, then
a \defn{$\gamma$-approximate $(r,s)$-clique core number} of $R$ is a value that is at least $k_R$ and at most $\gamma k_R$, where $\gamma > 1$.

An \defn{$a$-orientation} of an undirected graph is a total ordering on the vertices such that when edges in the graph are directed from vertices lower in the ordering to vertices higher in the ordering, the out-degree of each vertex is bounded by $a$.
Shi \textit{et al.} and Besta et al.~\cite{Besta2020} provide parallel work-efficient algorithms for computing an $O(\alpha)$-orientation, which take $O(m)$ work and $O(\log^2 n)$ span.

A \defn{connected component} in an undirected graph is a maximal subgraph such that all vertices in the subgraph are reachable from one another. Computing the connected components in a graph can be done in \BigO{m} work and \BigO{\log n} span \whp{}~\cite{Gazit1991}.

A \defn{union-find} data structure is used to represent collections of sets and supports the following two operations: \defn{unite($x$,$y$)} joins the sets that contain $x$ and $y$ and \defn{find($x$)} returns the identifier of the set containing $x$. Trees are used to implement union-find data structures, where elements have a parent pointer, and the root of a tree corresponds to the representative of the set.
Path compression can be done during unite and find operations to shortcut the pointers of elements traversed to point to the root of the set. We use the concurrent union-find data structure by Jayanti et al.~\cite{Jayanti2016}.

We use the notation $[n]$ to refer to the range of integers $[1,\ldots,n]$.

\myparagraph{Parallel Primitives} We use the following parallel primitives in our algorithms. 
We  use \defn{parallel hash tables}, which support $n$ insertion, deletion, and membership queries, in $O(n)$
work and $O(\log n)$ span \whp{}~\cite{gil91a}.
  \defn{List ranking} takes a linked list as input and returns the distance of each element to the end of the linked list. For a linked list of $n$ elements, list ranking can be solved in \BigO{n}
work and \BigO{\log n} span~\cite{JaJa92}. In our algorithms, we use list ranking to compute a unique identifier for each element in a linked list so that we can write them to an array in parallel.

\defn{Compare-and-swap$(x,old,new)$} attempts to atomically write the value $new$ into memory location $x$ if $x$ current stores the value $old$; it returns true if the write succeeds and false otherwise (meaning the value of $x$ was changed by another thread after it was read into $old$). We assume compare-and-swaps take $\BigO{1}$ work and span.

A \defn{parallel bucketing structure}
maintains a mapping from identifiers to buckets, which we use to group
$r$-cliques according to their incident $s$-clique counts~\cite{DhBlSh17}. The
buckets can change, and the structure can efficiently update these buckets. 
We take identifiers to be values associated with
$r$-cliques, and use the structure to repeatedly extract all $r$-cliques in
the minimum bucket, which can cause the buckets of other $r$-cliques to change (other $r$-cliques sharing vertices with extracted $r$-cliques in our algorithm).

\section{Overview of Contributions}

We present several novel algorithms for efficient parallel hierarchy construction for $(r, s)$ nucleus decomposition.
The prior state-of-the-art work on parallel $(r, s)$ nucleus decomposition, \ournd{}~\cite{shi2021nucleus}, only computes the $(r, s)$-clique core numbers and does not construct the hierarchy, which is non-trivial to accomplish. 
We first present in Section~\ref{sec:alg} a new theoretically efficient parallel hierarchy construction algorithm, \ourndh{}, that given the $(r, s)$-clique core numbers from~\cite{shi2021nucleus}, constructs the full hierarchy. In this two-phase algorithm, the hierarchy construction is entirely separate from the computation of the core numbers. 
The main innovation is a clever technique to store subgraphs corresponding to each $(r, s)$-clique core, such that the each subgraph only needs to be materialized when processing the corresponding level in the hierarchy tree. This limits the number of times we must iterate over each $r$-clique, and allows us to use a linear-work connectivity subroutine, giving us our theoretically efficient bounds.

We also present in Section~\ref{sec:approx} a new parallel approximation algorithm, \ourapproxnd{}, that computes approximate $(r, s)$-clique core numbers in polylogarithmic span without increasing the work. 
Notably, our parallel approximation algorithm is the first to achieve the dual guarantees of work-efficiency and polylogarithmic span.
We then combine \ourapproxnd{} with our theoretically efficient hierarchy construction subroutine, to form an approximate hierarchy algorithm, \ourapprox{}.

We observe that it is often not practically efficient to conduct a two-phase hierarchy algorithm (separating the hierarchy construction from the computation of the core numbers), and it is instead faster to construct the hierarchy directly while computing the core numbers. 
However, this presents a new series of challenges, where core numbers obtained later in the algorithm given by~\cite{shi2021nucleus} may affect all levels of the hierarchy tree, causing significant global changes in the tree structure as core numbers are computed. 
We present practical solutions in Section~\ref{sec:practical} limiting the cascading effects of these changes by compactly storing the tree using two simple data structures, a union-find structure and a hash table, and grouping $r$-cliques on-the-fly to reduce the propagating changes. 
We then introduce an efficient parallel post-processing step, to explicitly construct the final hierarchy tree from the two data structures.

Finally, we implement all of our algorithms, and present a comprehensive experimental evaluation in Section~\ref{sec:eval}.

\section{Nucleus Decomposition Hierarchy}
\label{sec:alg}

\begin{algorithm}[!t]
  \footnotesize
 \begin{algorithmic}[1]
 \State Initialize $r$, $s$ \Comment{$r$ and $s$ for $(r,s)$ nucleus decomposition}

\Procedure {\ourndh{}}{$G = (V,E)$}
\State \algblue{$ND \leftarrow$ \ournd{}($G$)} \Comment{Compute the nucleus core numbers, where $ND$ maps $r$-cliques to their core numbers} \label{alg-ndh:init-nd}
\State $k \leftarrow$ maximum core number in $ND$
\State For each $i \in [k]$, let $L_i$ denote a hash table, where the keys are $r$-cliques and the values are linked lists \label{alg-ndh:init-hash}
  \ParFor{\algblue{all $s$-cliques $S$ in $G$}} \label{alg-ndh:start-init-hash}
    \ParFor{all pairs of $r$-cliques $R, R'$ in $S$ where $ND[R'] \leq ND[R]$}\label{alg-ndh:begin-comb}
        \State Add $R'$ to the linked list keyed by $R$ in $L_{ND[R']}$ \label{alg-ndh:comb-construct}
    \EndParFor\label{alg-ndh:end-comb}
  \EndParFor\label{alg-ndh:end-init-hash}
\State Initialize the hierarchy tree $T$ with leaves corresponding to each $r$-clique \label{alg-ndh:init-tree}
\State For each $i \in [k]$, let $ID_i$ denote a hash table, where the keys are $r$-cliques and the values are $r$-cliques \label{alg-ndh:alloc-id}
\State Initialize each $ID_i$ to contain each key in $L_i$, mapped to itself \label{alg-ndh:init-id}
\For{$i \in \{ k, k-1, \ldots, 1\}$} \label{alg-ndh:start-big-loop}
  \State 
  Relabel each $r$-clique in each linked list in $L_i$ with its corresponding value in $ID_i$
  \label{alg-ndh:map-id}
  \State Apply parallel list ranking to transform the linked lists in $L_i$ into arrays, which forms a graph $H$ (where the edges are given by each key paired with each element in its corresponding linked list) \label{alg-ndh:list-rank}
  \State Run parallel linear-work connectivity on $H$ \label{alg-ndh:connect}
  \ParFor{each connected component $\mathcal{C} = \{ R_1, \ldots, R_c\}$ in $H$ where $|\mathcal{C}| \geq 2$}
    \State Construct a new parent in $T$ for all of the nodes corresponding to each $R_\ell$ (for $\ell \in [c]$), and represent the parent as the $r$-clique $R_1$ \label{alg-ndh:new-parent}
    \ParFor{each $j < i$ } \label{alg-ndh:start-ji-loop}
      \State Concatenate in parallel the linked lists corresponding to all $R_\ell$ (for $\ell \in [c]$) in $L_j$, as the updated value for the key $R_1$ in $L_j$ \label{alg-ndh:concat}
      \State For each $R_\ell$ (for $\ell \in [c]$), update the value of $R_\ell$ in $ID_j$ to be $R_1$ \label{alg-ndh:update-id}
    \EndParFor \label{alg-ndh:end-ji-loop}
  \EndParFor
\EndFor \label{alg-ndh:end-big-loop}
\State \Return $T$ \label{alg-ndh:return}
  \EndProcedure
 \end{algorithmic}
\caption{Parallel $(r,s)$ nucleus hierarchy algorithm}
 \label{alg:ndh}
\end{algorithm}

We now describe our theoretically efficient parallel nucleus decomposition hierarchy algorithm, \ourndh{}. \ourndh{} computes the hierarchy by first running an efficient parallel nucleus decomposition algorithm, \ournd{}~\cite{shi2021nucleus}, in order to obtain the $(r, s)$-clique core numbers corresponding to each $r$-clique. It then constructs a data structure consisting of $k$ levels, where $k$ is the maximum $(r, s)$-clique core number. Each level is represented by a hash table mapping sets of $r$-cliques to a linked list of $r$-cliques, which is used to efficiently store $s$-clique-adjacent $r$-cliques. \ourndh{} proceeds in levels through this data structure to construct the hierarchy tree from the bottom up, by performing linear-work connectivity on each level and updating the connectivity information in prior levels as a result.

\begin{figure}[t]
    \centering
   \begin{subfigure}{.4\textwidth}
   \centering
   \includegraphics[width=\columnwidth]{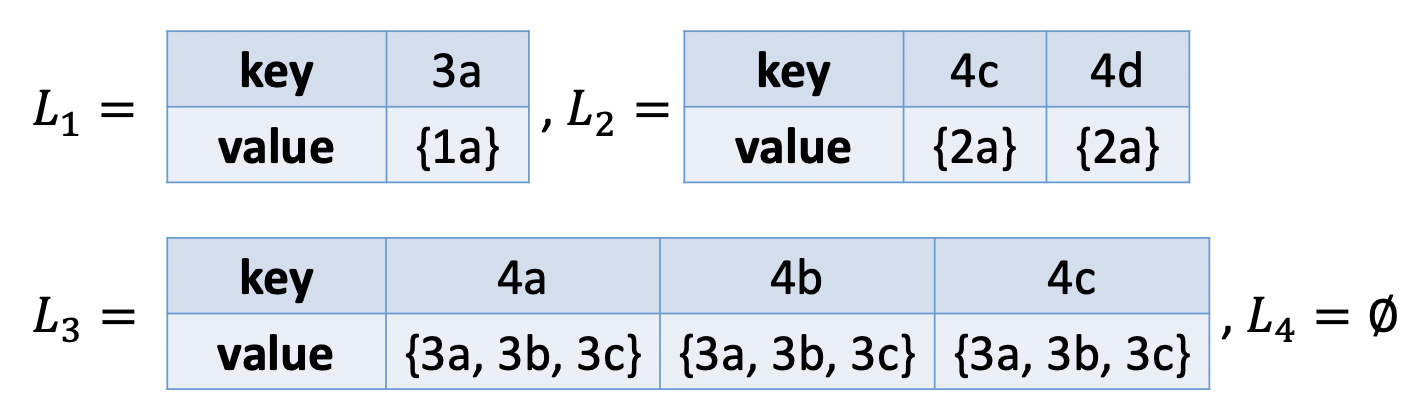}
   \end{subfigure}%
   \qquad
   \begin{subfigure}{.5\textwidth}
   \centering
   \includegraphics[width=\columnwidth]{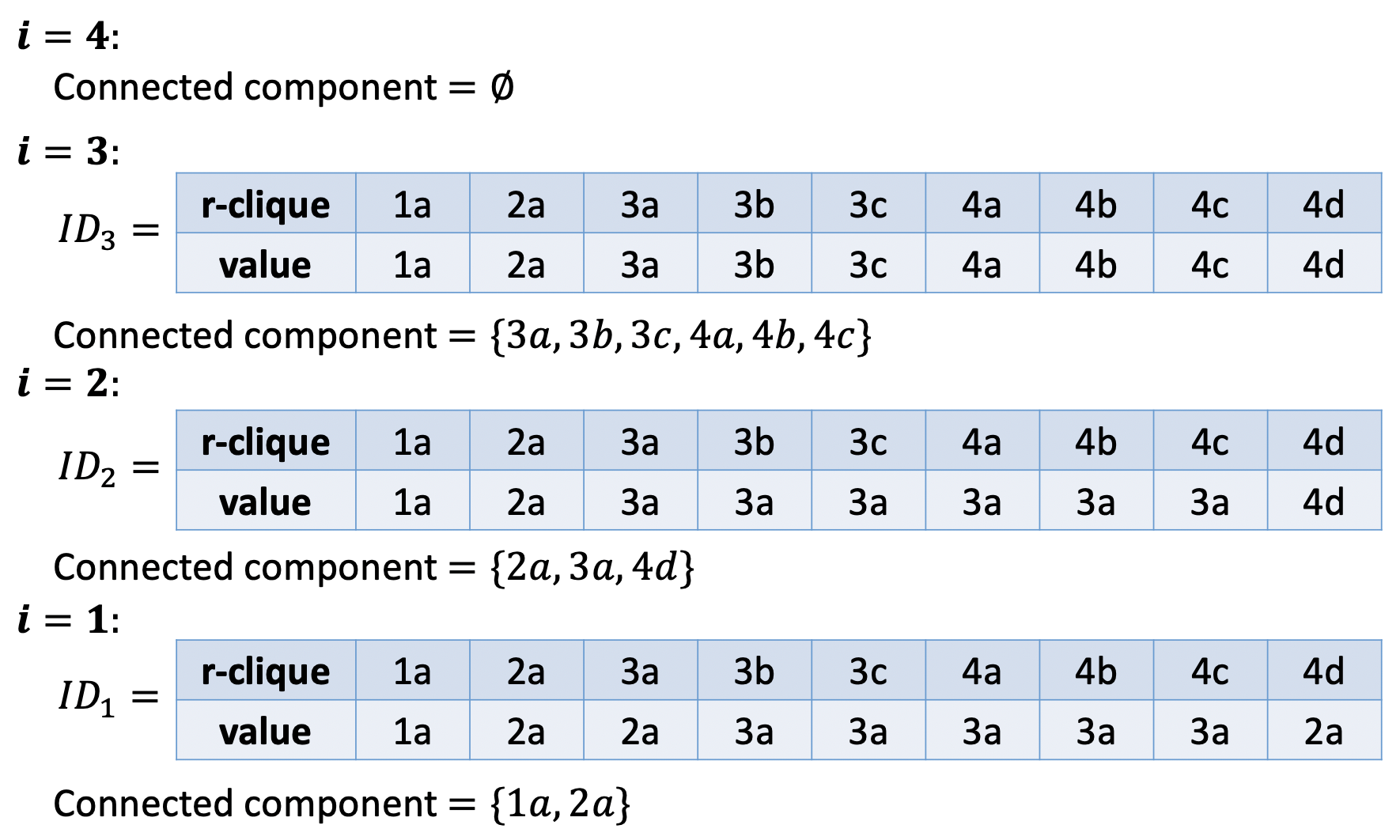}
      \end{subfigure}
      \vspace{-8pt}
   \caption{An example of the $L_i$ data structures maintained by \ourndh{} while computing the $(1,3)$-nucleus hierarchy on the graph in Figure~\ref{fig:graph}. For each round $i$ of the hierarchy construction, the connected components of $H$ and the $ID_{i}$ table used to construct $H$ is shown (except $ID_4$, which maps every $r$-clique to itself). 
   }
    \label{fig:alg1}
\end{figure}

\begin{figure}[t]
\vspace{-5pt}
    \centering
   \includegraphics[width=0.6\columnwidth]{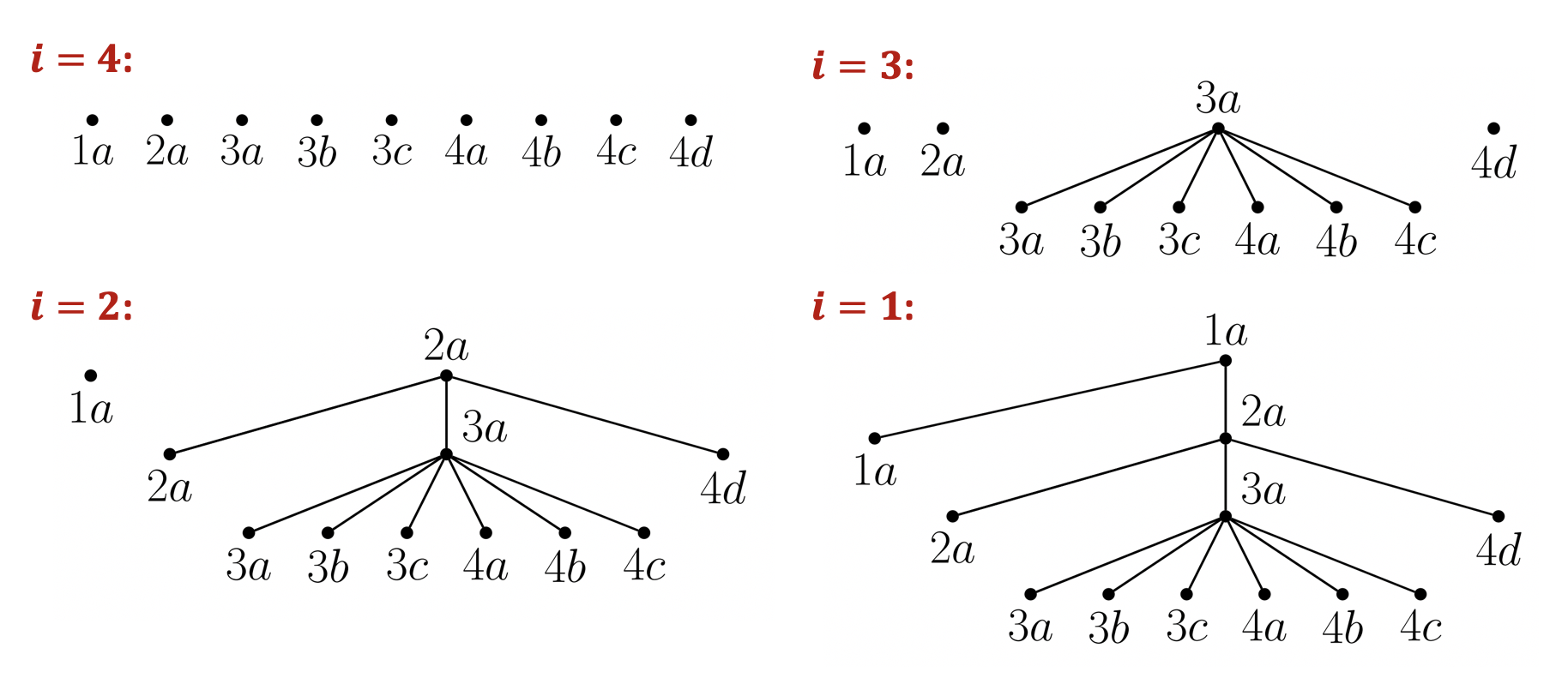}
   \vspace{-18pt}
   \caption{An example of the intermediate hierarchy trees $T$ after each round $i$ in \ourndh{}, while computing the $(1,3)$-nucleus hierarchy on the graph in Figure~\ref{fig:graph}.
   }
   \vspace{-1pt}
    \label{fig:alg1-tree}
\end{figure}

\myparagraph{Our Algorithm} We now provide a more detailed description of our algorithm. The pseudocode is in Algorithm \ref{alg:ndh}. 
Note that we have highlighted in blue the parts that derive directly from prior work; the remaining pseudocode is novel to this paper.
We refer to the example of $(1, 3)$-nucleus decomposition in Figure~\ref{fig:graph}. For simplicity, we have omitted labeling some of the vertices in the graph, and do not include them in our discussion.
The number in front of each vertex label is its core number. 

\ourndh{} first calls \ournd{} on Line~\ref{alg-ndh:init-nd}, to compute the $(r, s)$-clique core numbers of each $r$-clique. 
Note that this is the only instance where \ourndh{} uses a subroutine from the prior work~\cite{shi2021nucleus}, and the rest of the pseudocode is novel.
It stores these values in a hash table, $ND$, keyed by the $r$-cliques.
Then, it constructs a data structure consisting of $k$ hash tables on Line~\ref{alg-ndh:init-hash}, where $k$ is the maximum $(r,s)$-clique core number. Each hash table, $L_i$, maps $r$-cliques to a linked list of $r$-cliques, where $i$ corresponds to a core number. 
The first stage of \ourndh{} inserts all $s$-clique-adjacent $r$-cliques into the hash tables on Lines~\ref{alg-ndh:start-init-hash}--\ref{alg-ndh:end-init-hash}. Specifically, for adjacent $r$-cliques $R$ and $R'$ where $ND[R'] \leq ND[R]$, we insert $R'$ into the hash table corresponding to $R'$'s core number, $L_{ND[R']}$, with $\{R\}$ as the key. If an entry for $\{R\}$ already exists in $L_{ND[R']}$, we append $R'$ to the existing linked list. 
The hash tables are shown in Figure~\ref{fig:alg1}. For instance, 
$R' = 1a$ is adjacent to $R = 3a$, so we add $1a$ to the linked list keyed by $3a$ in $L_1$. 
In order to iterate in parallel over all $r$-cliques and over all $s$-cliques containing each $r$-clique, our algorithm uses a $s$-clique enumeration algorithm based on previous work by Shi \textit{et al.}~\cite{shi2020parallel}, which recursively finds and lists $c$-cliques in parallel and which can efficiently extend given $r$-cliques to find the $s$-cliques they are contained within.
\ourndh{} then takes all combinations of $r$ vertices in each discovered $s$-clique to find adjacent $r$-cliques in Lines~\ref{alg-ndh:begin-comb}--\ref{alg-ndh:end-comb}. 
Note that this $s$-clique enumeration subroutine is already used in \ournd{} in order to compute the $(r, s)$-clique core numbers of each $r$-clique, and in practice, we construct the hash tables $L_i$ while computing the core numbers in \ournd{} (rather than running the $s$-clique enumeration subroutine twice).

After constructing the initial set of hash tables, \ourndh{} proceeds to construct the nucleus decomposition hierarchy on Lines~\ref{alg-ndh:init-tree}--\ref{alg-ndh:update-id}. 
The broad idea is to construct the hierarchy starting at the leaf nodes, each of which correspond to an $r$-clique, and to merge tree nodes from the bottom up (i.e., in decreasing order of core number).
The algorithm begins by considering only connected components formed by $r$-cliques with the greatest core number, $k$, which dictate the leaf nodes to be merged into super-nodes at the second-to-last level of the hierarchy tree based on their $s$-clique connectivity.
In subsequent rounds, when processing core number $i$, the algorithm considers connected components formed by $r$-cliques with core numbers $\geq i$, which similarly dictate the merges to be performed in the next level of the hierarchy tree. 
\ourndh{} efficiently maintains connected components from higher core numbers when computing connected components at lower core numbers, by concatenating the relevant linked lists in the hash tables corresponding to the lower core numbers when processing higher core numbers. 
Figure~\ref{fig:alg1-tree} shows the construction of this hierarchy throughout the different rounds (for $i = 4, \ldots, 1$), and Figure~\ref{fig:alg1} lists the connected components computed in each round.
In our example, since we omit the other vertices in
the $4$-$(1,3)$-nucleus besides $4a$, $4b$, $4c$, and $4d$ (for simplicity), 
we begin with singleton leaf nodes after processing round $i = 4$.

In more detail, on Line~\ref{alg-ndh:init-tree}, we begin with a hierarchy tree $T$ consisting only of leaf nodes corresponding to each $r$-clique. We also initialize a data structure consisting of $k$ hash tables on Lines~\ref{alg-ndh:alloc-id}--\ref{alg-ndh:init-id}, where each hash table $ID_i$ maps $r$-cliques to $r$-cliques. The idea is to maintain a mapping of each $r$-clique to the component that it is contained within in $L_i$, for the corresponding level. Initially, each $ID_i$ maps each key in $L_i$ to itself.

Then, \ourndh{} considers each core number $i$, starting from $k$ and going down to $1$ (Line~\ref{alg-ndh:start-big-loop}). For each $i$, we relabel each $r$-clique in each linked list in $L_i$ with its corresponding value in $ID_i$ (Line~\ref{alg-ndh:map-id}). Then, it 
uses parallel list ranking to convert all linked lists in $L_i$ to arrays, forming a graph $H$, where 
edges are given by each key paired with each element in its corresponding linked list (Line~\ref{alg-ndh:list-rank}). 
Note that the edges actually denote $s$-clique-adjacent $r$-cliques (or components of $r$-cliques). The vertices correspond to $r$-cliques, which represent components of $r$-cliques in $G$. In our example, we note that there is nothing to be done for round $i = 4$, since $L_4$ is empty.
Therefore, we begin by processing the hash table $L_3$ in round $i = 3$. $ID_3$ maps every vertex to itself, so we do not relabel any of the labels in $L_3$. The graph $H$ that we construct consists of the vertices $3a$, $3b$, $3c$, $4a$, $4b$, and $4c$, with edges given by $\{3a, 3b, 3c\} \times \{ 4a, 4b, 4c\}$.

\ourndh{} proceeds by running a linear-work parallel connectivity algorithm on $H$ (Line~\ref{alg-ndh:connect}), and processes the given connected components. 
Note that each connected component $C$ consists of a set of vertices in $H$, which represents a set of $r$-cliques, say $\{R_1, \ldots, R_c\}$.
In our example, in round $i = 3$, we have a single connected component in $H$ consisting of all of the vertices, $3a$, $3b$, $3c$, $4a$, $4b$, and $4c$.
Then, for each such connected component representing more than one $r$-clique, in the hierarchy tree $T$, we construct a new parent for the nodes corresponding to each $R_\ell$ for $\ell \in [c]$ (Line~\ref{alg-ndh:new-parent}). 
Note that each $R_\ell$ could correspond to a subtree containing multiple tree nodes, owing to parent nodes constructed in previous steps.
We represent the new parent by the $r$-clique $R_1$, which we designate arbitrarily as the representative for the connected component $C$.
In Figure~\ref{fig:alg1-tree}, under $i = 3$, we construct a new parent node labeled by $3a$, which we designate as the representative of the component $\{3a, 3b, 3c, 4a, 4b, 4c\}$.

Finally, for each core number $j < i$, we update the connectivity information on each hash table $L_j$, by concatenating all linked lists in $L_j$ corresponding to the $r$-cliques in the component $C$ (Line~\ref{alg-ndh:concat}). More precisely, we update the value of the key $R_1$ to be the concatenation, or we insert the concatenation into $L_j$ with $R_1$ as the key if $R_1$ does not already exist in the hash table (this is possible if $R_1$ had no neighbors with core number $j$).  
We use tombstones to delete the other keys $R_\ell$ ($\ell \neq 1$) in each $L_j$. 
Additionally, on Line~\ref{alg-ndh:update-id}, we update the component ID in $ID_j$ of each $R_\ell$ to be $R_1$.
For $i=3$, we see that in our example there are no lists to concatenate, but in $ID_2$, we map each of $3b, 3c, 4a, 4b$, and $4c$ to the representative $3a$.

We repeat this process until we have processed all $k$ rounds. 
In our example, for round $i = 2$, we relabel $4c$ in $L_2$ with $3a$, as given by $ID_2$, and we construct a subgraph $H$ with vertices $2a$, $3a$, and $4d$, and edges $(2a, 3a)$ and $(2a, 4d)$. Again, there is only one connected component, given by $\{2a, 3a, 4d\}$. In the hierarchy tree for round $i = 2$, we construct a new parent labeled with the representative $2a$, whose children are the leaves $2a$ and $4d$, and the previously constructed parent node $3a$. Again, there are no lists to concatenate in $L_1$, but in $ID_1$, we map both $3a$ and $4d$ to the representative $2a$.
Then, for round $i=1$, we relabel $3a$ in $L_1$ with $2a$, as given by $ID_1$. We construct a subgraph $H$ with vertices $1a$ and $2a$, and a single edge between them. There is only one connected component, $\{1a, 2a\}$, and in the hierarchy tree for round $i=1$, we construct a new parent labeled with the representative $1a$, whose children are the leaf $1a$ and the previously constructed parent node $2a$. This concludes the construction of the hierarchy $T$ in our example.

\myparagraph{Theoretical Efficiency} We now analyze the theoretical efficiency of our hierarchy algorithm, \ourndh{}. Note that as in Shi \textit{et al.}'s~\cite{shi2021nucleus} work, $\rho_{(r,s)}(G)$ is defined to be the \defn{$\boldsymbol{(r,s)}$ peeling complexity} of
$G$, or the number of rounds needed to peel the graph where in each
round, all $r$-cliques with the minimum $s$-clique count are peeled (removed).  Importantly, $k \leq \rho_{(r,s)}(G) \leq O(m\alpha^{r-2})$, since at least
one $r$-clique is peeled in each round, and the number of rounds is at least the maximum $(r, s)$-clique core number in $G$.

\vspace{-3pt}
\begin{theorem}\label{thm:ndh}
  \ourndh{} computes the $(r,s)$ nucleus decomposition hierarchy in $\BigO{m\alpha^{s-2} }$ expected work and $\BigO{k \log n +\allowbreak \rho_{(r,s)}(G) \log n +\allowbreak \log^2 n}$ span \whp{}, where
  $\rho_{(r,s)}(G)$ is the $(r,s)$ peeling complexity and $k$ is the maximum $(r, s)$-clique core number.
\end{theorem}
\vspace{-3pt}
\begin{proof}
First, the theoretical complexity of computing the $(r, s)$-clique core numbers of each $r$-clique (Line~\ref{alg-ndh:init-nd}) is given by Shi \textit{et al.}~\cite{shi2021nucleus}, which they show takes $\BigO{m\alpha^{s-2}}$ work and $\BigO{\rho_{(r,s)}(G) \log n \allowbreak + \log^2 n}$ span \whp{} for constant $r$ and $s$. Note that the work bound is given by the version of their algorithm that takes space proportional to the number of $s$-cliques in $G$, which we incur regardless to store the hash tables $L_i$.
Additionally, iterating through all $s$-cliques in $G$ (Line~\ref{alg-ndh:start-init-hash}) is superseded by the work required to compute the $(r, s)$-clique core numbers.
For every pair of $r$-cliques in each $s$-cliques, we hash each $r$-clique and append to a linked list (Lines~\ref{alg-ndh:begin-comb}--\ref{alg-ndh:end-comb}), which in total takes 
$O(m\alpha^{s-2})$ work and $O(\log^2 n)$ span \whp{} for constant $r$ and $s$.

We now discuss the work and span of constructing the hierarchy tree $T$ level-by-level, in $k$ rounds (Lines~\ref{alg-ndh:start-big-loop}--\ref{alg-ndh:end-big-loop}). 
The key idea here is that the linked lists in each hash table $L_i$ are iterated over at most once across all rounds, and the concatenation of linked lists in intermediate rounds incurs minimal costs (since concatenating linked lists does not require iterating through the linked lists).
The sum of the lengths of the linked lists in each $L_i$ remains invariant throughout this portion of the construction, so in total, the cost of iterating over the linked lists is bounded by $O(m\alpha^{s-2})$ work, matching the work needed to construct each linked list in each $L_i$ originally. Also, the number of keys in each $L_i$ only monotonically decreases, so processing the connected components of the constructed subgraphs $H$ takes at most $O(m\alpha^{s-2})$ work as well.

In more detail, for fixed $i$, let the sum of the lengths of the linked lists in $L_i$ be $\ell_i$, and let the number of keys in $L_i$ be $y_i$.
For each round $i \in [k]$, applying $ID_i$ and parallel list ranking on the linked lists in $L_i$ (Lines~\ref{alg-ndh:map-id}--\ref{alg-ndh:list-rank}) takes work proportional to $\ell_i$ and $\BigO{\log n}$ span \whp{}. The number of edges in the subgraph $H$ constructed from the linked lists in $L_i$ is at most $\ell_i$, so performing connectivity on $H$ (Line~\ref{alg-ndh:connect}) takes work linear in $\ell_i$ and $\BigO{\log n}$ span \whp{}~\cite{Gazit1991}. 
Also, there are at most $\BigO{y_i}$ vertices in $H$, so processing each connected component and updating the hierarchy tree $T$ (Line~\ref{alg-ndh:new-parent}) takes $\BigO{y_i}$ work and $\BigO{1}$ span. In total, for Lines~\ref{alg-ndh:map-id}--\ref{alg-ndh:new-parent}, we incur $\BigO{\sum_i (\ell_i + y_i) } = \BigO{m\alpha^{s-2} }$ expected work and $\BigO{k \log n}$ span \whp{}

It remains to bound the cost of of the loop in Lines~\ref{alg-ndh:start-ji-loop}--\ref{alg-ndh:end-ji-loop}.
First, note that across all rounds, each value corresponding to every key in each $L_i$ is concatenated at most once (Line~\ref{alg-ndh:concat}). 
This is because we concatenate linked lists and store the concatenation under exactly one existing key. 
We empty the corresponding values for the previously associated keys, which allows us to maintain the invariant that the sum of the lengths of the linked lists in each $L_i$ remains fixed.
Also, the previously associated keys are never used again, since we update the value associated with each $r$-clique in $ID_j$. 
Thus, the concatenations are bounded by $O(m\alpha^{s-2})$ work across all rounds, matching the work of constructing all $L_i$. 
Additionally, iterating through all levels $j \leq i$ for each component $\mathcal{C}$ does not increase the work. This is because we only reach this loop if $|\mathcal{C}| \geq 2$, which means that we are effectively merging multiple $r$-cliques together in each hash table $L_j$ for $j \leq i$, and assigning the $r$-clique $R_1$ as the new representative for the other $r$-cliques in the component (updating $ID_j$).
Thus, we can assign the work of iterating through all $j \leq i$ to the $r$-cliques that are being merged ($R_\ell$ where $\ell \neq 1$). 
Once the $r$-clique $R_\ell$ is merged, due to the concatenation on Line~\ref{alg-ndh:concat} and the update in $ID_j$ on Line~\ref{alg-ndh:update-id}, it never participates as a vertex in $H$ again in future rounds, so is never re-processed in a future connected component of $H$; this is due to the mapping on Line~\ref{alg-ndh:map-id}. 
The amount of work that we assign per merged $r$-clique $R_\ell$ is at most the core number of $R_\ell$. This is because the amount of work we assign to $R_\ell$ is given by the number of rounds in which we merge, or $i$, and $R_\ell$ only appears in $L_i$ if $ND[R_\ell] \geq i$, by construction on Line~\ref{alg-ndh:comb-construct}. 
Thus, in total, for each $r$-clique, we incur work upper bounded by the core number.
We have in total $\BigO{\sum_{r\text{-clique }R \in G} ND[R]} = \BigO{m\alpha^{s-2}}$ work, since the sum of the $(r, s)$-clique core numbers in $G$ is necessarily bounded by the number of $s$-cliques for constant $r$ and $s$. This is because each $s$-clique contributes to at most $\binom{s}{r}$ $r$-clique's core numbers, so the summation across all core numbers is upper bounded by $\binom{s}{r} \cdot $ (the number of $s$-cliques).
Thus, in total, the loop in Lines~\ref{alg-ndh:start-ji-loop}--\ref{alg-ndh:end-ji-loop} incurs $\BigO{m\alpha^{s-2}}$ work and $\BigO{k \log n}$ span, where the span is due to list ranking.

In total, we have $\BigO{m\alpha^{s-2} }$ expected work and $\BigO{k \log n \allowbreak +\allowbreak \rho_{(r,s)}(G) \log n \allowbreak +\allowbreak \log^2 n}$ span \whp{}, as desired.
\end{proof}

\myparagraph{Comparison to Prior Work} The prior state-of-the-art algorithm is the sequential $(r, s)$ nucleus decomposition hierarchy algorithm by \sariyuce{} and Pinar~\cite{SaPi16}. 
They provide an algorithm similar to \ourndh{} in that it first computes the $(r, s)$-clique core numbers of each $r$-clique, and then builds the hierarchy tree from the bottom up. 
They show that the time complexity is upper bounded by the time complexity of computing the $(r, s)$-clique core numbers.
Note that they omit a factor of $\BigO{\alpha(n_s,n_r)}$ where $\alpha$ is the inverse Ackermann function and $n_s$ and $n_r$ are the number of $s$-cliques and $r$-cliques, respectively, in the graph; this factor is necessary for their algorithm, since they use union-find to construct the hierarchy tree. 
Thus, the time complexity of \sariyuce{} and Pinar's algorithm is $\BigO{m\alpha^{s-2} \alpha(n_s,n_r)}$ (this is achieved using the state-of-the-art algorithm for computing the $(r, s)$-clique core numbers~\cite{shi2021nucleus}).
Our \ourndh{} algorithm avoids the additional inverse Ackermann factor, by efficiently constructing graphs to represent different levels of the hierarchy tree and using linear-work graph connectivity to construct the hierarchy tree.
Thus, we improve the sequential running time of constructing the $(r, s)$ nucleus decomposition hierarchy to $\BigO{m\alpha^{s-2}}$, and \ourndh{} is work-efficient.

Shi et al.~\cite{shi2021nucleus} present a parallel algorithm to compute the $(r,s)$-clique core numbers in $O(m\alpha^{s-2})$ expected work.
A vanilla extension of their algorithm to compute the hierarchy would be to run connected components for each non-empty core number $i$, with vertices being the $(r,s)$-cliques that have core number $i$ or above, in a top-down manner. This incurs $O(m\alpha^{s-2})$ expected work per level, and there can be $\rho_{(r,s)}(G)$ levels in the worst case, leading to $O(\rho_{(r,s)}(G)m\alpha^{s-2})$ total expected work. In contrast, our algorithm has a much lower total expected work of $O(m\alpha^{s-2})$.

\section{Approximate Nucleus Decomposition}\label{sec:approx}

Given the potentially large span (i.e., longest critical path) of computing the exact nucleus decomposition hierarchy, we develop a new parallel approximate nucleus decomposition hierarchy algorithm, \ourapprox{}. 
Instead of computing exact $(r, s)$-clique-core numbers of each $r$-clique, the main idea of the new algorithm is to compute an approximation of the $(r, s)$-clique core number. 
Specifically, we compute a $(\binom{s}{r} + \varepsilon)$-approximate $(r, s)$-clique core number of each $r$-clique; that is to say, if we let the true $(r,s)$-clique core number of each $r$-clique be $k_R$, our approximation is at least $k_R$ and at most $(\binom{s}{r} + \varepsilon) \cdot k_R$.

Our approximate computation uses the same peeling paradigm from Shi \textit{et al.}~\cite{shi2021nucleus} for exact nucleus decomposition, but with an important modification that allows it to take only $\BigO{\log^2 n}$ peeling rounds, thus significantly improving upon the span. 
As a result, we only have polylogarithmically many core numbers, leading to a hierarchy tree with polylogarithmic height. 
Our hierarchy construction for \ourapprox{} is exactly the same as that of \ourndh{}, and the only salient difference is that for \ourapprox{}, we replace the \ournd{} subroutine in Line~\ref{alg-ndh:init-nd} of Algorithm~\ref{alg:ndh} with our approximate nucleus decomposition subroutine, \ourapproxnd{}.

\myparagraph{Our Algorithm} 
We present our pseudocode for \ourapproxnd{} in Algorithm~\ref{alg:ndh-approx}. 
Again, we have highlighted in blue the parts that derive directly from prior work, and the remainder is novel to this work.
Note that it takes as input a parameter $\delta > 0$, which controls the $\varepsilon$ in the $(\binom{s}{r} + \varepsilon)$-approximation. Specifically, \ourapproxnd{} gives a $(\binom{s}{r} + \delta) \cdot (1 + \delta ) = (\binom{s}{r}+\varepsilon)$-approximation, which we prove in Theorem~\ref{thm:approx}.

First, on Line~\ref{alg-approxnd:orient}, \ourapproxnd{} computes a low out-degree orientation of the graph $G$, which directs the edges such that every vertex has out-degree at most $\BigO{\alpha}$, using an efficient algorithm by Shi \textit{et al.}~\cite{shi2020parallel}. Then, on Lines~\ref{alg-approxnd:count}--\ref{alg-approxnd:count-end}, it counts the number of $s$-cliques per $r$-clique in $G$, and stores the result in a parallel hash table $U$. It uses an $s$-clique counting subroutine, \ourcountrec{}, from Shi \textit{et al.}'s previous work~\cite{shi2020parallel}.
The key difference between \ourapproxnd{} and \ournd{} is on Line~\ref{alg-approxnd:buckets}, where the buckets in the bucketing structure hold $r$-cliques with a range of $s$-clique-degrees instead of a single $s$-clique-degree. 
Specifically, for an input parameter $\delta$, we define the range of each bucket $B_i$ to be $[(\binom{s}{r} + \delta) \cdot (1+\delta)^i , (\binom{s}{r} + \delta)\cdot (1+\delta)^{i+1}]$, where $i \in [s \log_{1 + \delta} n]$ since $\binom{n}{s} = \BigO{n^s}$ is a trivial upper bound on the maximum $s$-clique-degree possible in any given graph.

The peeling algorithm then proceeds as it does in~\cite{shi2021nucleus}, except using our modified bucketing structure. 
While not all $r$-cliques have been peeled, \ourapproxnd{} processes the set of $r$-cliques $A$ (that have not yet been peeled) within the lowest bucket $B_i$ (starting with $i = 0$), and peels them from the graph (Lines~\ref{alg-approxnd:start-peel}--\ref{alg-approxnd:end-peel}).
For each $r$-clique $R$ in $A$, we iterate over all $s$-clique-adjacent $r$-cliques $R'$, and update the recorded $s$-clique-degree of $R'$ given $R$'s removal (Lines~\ref{alg-approxnd:start-count}--\ref{alg-approxnd:end-count}). 
Then, we peel (remove) the $r$-cliques in $A$ from the graph and update the buckets of all unpeeled $r$-cliques based on the updated $s$-clique-degrees on Line~\ref{alg-approxnd:update-bucket}. Notably, if a $r$-clique $R$'s $s$-clique-degree falls below the range of the current bucket of $r$-cliques that is being peeled, we do not rebucket $R$ into a lower bucket, and instead aggregate these $r$-cliques within the current bucket. 
As such, in any given round of peeling, we are actually peeling all $r$-cliques with $s$-clique-degree $\leq (\binom{s}{r} + \delta) \cdot(1+\delta)^{i}$, which is important for our theoretical bounds.
Note that we process a given bucket $B_i$ at most $\BigO{\log_{1+\slfrac{\delta}{\binom{s}{r}}}(n)}$ times; if we have exceeded this threshold, or if $B_i$ is empty, then we move on to processing the next bucket, $B_{i+1}$ (Lines~\ref{alg-approxnd:bound-rounds}--\ref{alg-approxnd:end-bound-rounds}). If there are unpeeled $r$-cliques remaining in $B_i$ once we reach this threshold, we include them in the next bucket $B_{i+1}$ (Line~\ref{alg-approxnd:carry-over}).
Note that the approximate $(r, s)$-clique core number that we compute for each $r$-clique is given by the upper bound of the bucket in which it was peeled. In practice, we can improve this by taking the minimum of the upper bound of the bucket, and the $s$-clique-degree of each $r$-clique (in the original graph).

Once we have peeled all $r$-cliques, this concludes our subroutine. \ourapprox{} is then given by replacing \ournd{} in Line~\ref{alg-ndh:init-nd} of Algorithm~\ref{alg:ndh} with \ourapproxnd{}.

\begin{algorithm}[!t]
  \footnotesize
 \begin{algorithmic}[1]
 \State Initialize $r$, $s$ \Comment{$r$ and $s$ for $(r,s)$ nucleus decomposition}
\Procedure {\ourapproxnd{}}{$G = (V,E)$, $\delta$}
\State \algblue{$DG \leftarrow$ \algname{Arb-Orient}($G$)} \Comment{Apply an arboricity-orientation algorithm} \label{alg-approxnd:orient}
\State Initialize $U$ to be a parallel hash table with $r$-cliques as keys, and $s$-clique counts as values \label{alg-approxnd:count}
\State \algblue{\ourcountrec{}($DG$, $s$, $U$)} \Comment{Count $s$-cliques, and store the counts per $r$-clique in $U$} \label{alg-approxnd:count-end}
\State Let $ND$ be a bucketing structure mapping each $r$-clique to a bucket based on \# of $s$-cliques, where each bucket $B_i$ contain all $r$-cliques with $s$-clique-degree in the range $[(\binom{s}{r} + \delta) \cdot (1 + \delta)^i , (\binom{s}{r} + \delta)\cdot (1 + \delta)^{i+1}]$, for all $i \in [s \log_{\binom{s}{r} + \delta} n]$\label{alg-approxnd:buckets}
\State $\mathsf{finished} \leftarrow 0$, $\mathsf{num\_rounds} \leftarrow 0$, $i \leftarrow 0$
\While{$\mathsf{finished} < |U|$} \label{alg-approxnd:start-peel}
\State $A \leftarrow$ $r$-cliques in the bucket $B_i$ in $ND$ (to be peeled)
\State $\mathsf{finished} \leftarrow \mathsf{finished} + |A|$
\State $\mathsf{num\_rounds} \leftarrow \mathsf{num\_rounds} + 1$
\ParFor{all $r$-cliques $R$ in $A$}  \label{alg-approxnd:start-count}
  \ParFor{\algblue{all $s$-cliques $S$ containing $R$}}
    \ParFor{all $r$-cliques $R'$ in $S$ where $R' \neq R$} 
        \State Update $s$-clique count of $R'$ in $U$
    \EndParFor
  \EndParFor
\EndParFor  \label{alg-approxnd:end-count}
\State Peel $A$, assign the 
upper bound of the bucket as the coreness value for peeled $r$-cliques,
and update the buckets of $r$-cliques with updated $s$-clique counts\label{alg-approxnd:update-bucket}
\If{$\mathsf{num\_rounds} \geq \BigO{\log_{1+\slfrac{\delta}{\binom{s}{r}}}(n)}$ or $B_i$ is empty} \label{alg-approxnd:bound-rounds}
  \State Add the remaining $r$-cliques in $B_i$ (if it is non-empty) to $B_{i+1}$ \label{alg-approxnd:carry-over}
  \State $i \leftarrow i + 1$
  \State $\mathsf{num\_rounds} \leftarrow 0$
\EndIf \label{alg-approxnd:end-bound-rounds}
\EndWhile \label{alg-approxnd:end-peel}
\State \Return $ND$
  \EndProcedure
 \end{algorithmic}
\caption{Approximate parallel $(r,s)$ nucleus algorithm}
 \label{alg:ndh-approx}
\end{algorithm}

\myparagraph{Theoretical Guarantees and Efficiency} We now discuss the theoretical guarantees and theoretical efficiency of \ourapproxnd{}, and by extension, \ourapprox{}. 

\iffull
We introduce the following lemmas to help prove that \ourapproxnd{} guarantees a $(\binom{s}{r} + \varepsilon)$-approximation of the true $(r, s)$-clique core numbers of each $r$-clique.

In particular, Lemma~\ref{lemma:1} bounds the number of $r$-cliques with $(r, s)$-clique core numbers $\leq \ell$ for a fixed $\ell$. 
We use this to prove Lemma~\ref{lemma:2}, which bounds the proportion of $r$-cliques with core numbers $\leq \ell$, but with $s$-clique-degree $> \ell (\binom{s}{r} + \delta)$. 
We can then set $\ell$ such that at any given step of our peeling process, we obtain a bound on the maximum proportion of $r$-cliques with core number at most $\ell$ that is not within the current bucket to be peeled. 
In essence, this gives us a bound on the number of times that a bucket must be reprocessed, such that moving on to the next bucket does not degrade the approximation factor, which gives us our approximation guarantees.

Note that Lemmas~\ref{lemma:1} and~\ref{lemma:2} apply to any stage of the peeling process in \ourapproxnd{}; that is to say, even if we have peeled a set of $r$-cliques from the graph $G$, the lemmas hold true for the remaining unpeeled $r$-cliques in $G$, with the updated $s$-clique-degrees 
(which are the original $s$-clique-degrees minus the number of incident $s$-cliques that have been removed from the graph due to the set of peeled $r$-cliques).

\begin{lemma} \label{lemma:1}
Let $S_\ell$ be the set of remaining, or unpeeled, $r$-cliques with $(r, s)$-clique core numbers $\leq \ell$, considering an arbitrary fixed stage in the peeling process. Then, the number of $s$-cliques incident to $S_\ell$ is $\leq \ell \cdot |S_\ell|$.
\end{lemma}
\begin{proof}
We prove this by considering the exact $(r, s)$ nucleus decomposition algorithm given by \sariyuce{} \textit{et al.}~\cite{Sariyuce2017}. In their algorithm, at each step, an $r$-clique with the minimum $s$-clique-degree in the graph is peeled.
In peeling an $r$-clique $R$, for every $s$-clique $S$ that $R$ participated in, the algorithm decrements the $s$-clique-degree of the remaining $r$-cliques $R' \in S$ that have not yet been peeled.
The $(r, s)$-clique core number of each $r$-clique $R$ is given by the maximum $k_R$ such that at some point before $R$ was peeled in this algorithm, all $r$-cliques that were not yet peeled had $s$-clique-degree at least $k_R$. 
In order for a given $r$-clique $R$ to have $(r, s)$-clique core number $k_R$, it must have had $s$-clique-degree at most $k_R$ when it was peeled from the graph. 
If $R$ has $s$-clique-degree greater than $k_R$ when it was peeled, since $R$ must have had the minimum $s$-clique-degree when it was peeled, this would mean that before $R$ was peeled, every $r$-clique had $s$-clique-degree greater than $k_R$, so by definition, $R$'s core number must be greater than $k_R$, which is a contradiction.

Note that \sariyuce{} \textit{et al.}~\cite{Sariyuce2017} prove that considering the $r$-cliques in the order in which they are peeled, the $(r, s)$-clique core numbers of the $r$-cliques are monotonically increasing.

Using these observations, the set $S_\ell$ of $r$-cliques is given by a contiguous sequence of $r$-cliques peeled in this algorithm, where upon being peeled, each $r$-clique has induced $s$-clique-degree at most $\ell$. As such, we can assign each $s$-clique $S$ incident to $S_\ell$ to the $r$-clique $R$ in which $S$ contributed to $R$'s induced $s$-clique-degree when $R$ was peeled.
As a result, the total number of $s$-cliques incident to $S_\ell$ must be $\leq \ell \cdot |S_\ell|$, since the induced $s$-clique-degree of each such $r$-clique when it was peeled is upper bounded by $\ell$.
\end{proof}

\begin{lemma} \label{lemma:2}
Let $F_\ell$ be the set of unpeeled, or remaining, $r$-cliques with $s$-clique-degree $> \ell (\binom{s}{r} + \delta)$ and with $(r, s)$-clique core number $\leq \ell$, considering an arbitrary fixed stage in the peeling process. Then, $|F_\ell| \leq \slfrac{\binom{s}{r} \cdot |S_\ell|}{(\binom{s}{r} + \delta)}$.
\end{lemma}
\begin{proof}
Let $S_\ell$ be the set of unpeeled $r$-cliques with $(r,s)$-clique core numbers $\leq \ell$. First, note that $F_\ell \subseteq S_\ell$ by definition. By Lemma~\ref{lemma:1}, we know that at most $\ell \cdot |S_\ell|$ $s$-cliques are incident to $S_\ell$. Thus, the sum of the $s$-clique-degrees of the $r$-cliques in $S_\ell$ is at most $\binom{s}{r} \cdot \ell \cdot |S_\ell|$, since each $s$-clique can contribute to the $s$-clique-degree of at most $\binom{s}{r}$ $r$-cliques. Then, since each $r$-clique in $F_\ell$ has $s$-clique-degree $> \ell(\binom{s}{r} + \delta)$ by definition, the maximum number of $r$-cliques in $F_\ell$ is $\slfrac{\binom{s}{r} \cdot \ell \cdot |S_\ell|}{(\ell(\binom{s}{r} + \delta))} = \slfrac{\binom{s}{r} \cdot |S_\ell|}{(\binom{s}{r} + \delta)}$, as desired.
\end{proof}

\fi

\begin{theorem}\label{thm:approx}
  \ourapprox{} computes a $(\binom{s}{r} + \varepsilon)$-approximate $(r,s)$ nucleus decomposition hierarchy in $\BigO{m\alpha^{s-2} }$ expected work and $\BigO{ \log^3 n}$ span \whp{}
\end{theorem}

\iffull

\begin{proof}

Note that the work bound for \ourapprox{} follows directly from that for \ournd{}~\cite{shi2021nucleus} and \ourndh{}. This is because the only difference between \ourapproxnd{} and \ournd{} is the boundaries used in the bucketing structure, which affects the number of $r$-cliques that are peeled in any given round, but does not affect the amount of work needed to rediscover $s$-cliques containing peeled $r$-cliques. The latter dominates the work of both \ourapproxnd{} and \ournd{}, so as a result, the work of \ourapproxnd{} is precisely the work of \ournd{}. To be more specific, as given by~\cite{shi2021nucleus}, \ourapproxnd{} takes $\BigO{m\alpha^{s-2} }$ work \whp{} assuming space proportional to the number of $s$-cliques.\footnote{
As an aside, following the arguments in~\cite{shi2021nucleus}, if we instead restrict our space usage to be proportional to the number of $r$-cliques, we can modify the bucketing structure to use a batch-parallel Fibonacci heap~\cite{Shi2020}, which would increase the work bound to $\BigO{m\alpha^{s-2} + \log^3 n}$ amortized expected work \whp{}}
Additionally, \ourapprox{} is precisely \ourndh{}, except using \ourapproxnd{} to compute the $(r, s)$-clique core numbers. Thus, following the proof of Theorem~\ref{thm:ndh}, we see that \ourapprox{} overall takes $\BigO{m\alpha^{s-2} }$ expected work \whp{}
The work bound is unaffected by the choice of $\epsilon$.

For our span bound, first note that there are only a logarithmic number of possible $i$ values (since $\binom{n}{s} = \BigO{n^s}$ upper bounds the maximum possible $s$-clique-degree). For each bucket $B_i$, we by construction process $B_i$ at most $\BigO{\log_{1+\slfrac{\delta}{\binom{s}{r}}}(n)}$ times (Line~\ref{alg-approxnd:bound-rounds}). 
Thus, in total, we have $\BigO{\log^2 n}$ rounds of peeling in \ourapproxnd{}. 
The span of each peeling round is $\BigO{\log n}$ \whp{} to retrieve the next bucket of $s$-cliques and to perform hash table operations to update the $s$-clique counts, as discussed in more detail by Shi \textit{et al.}~\cite{shi2021nucleus}. Therefore, the total span of peeling is $\BigO{\log^3 n}$ span \whp{}
Our choice of $\epsilon$ later will affect the base of one of the logarithms. 

The work and span of orienting the graph and counting the number of $s$-cliques follows from Shi \textit{et al.}'s $s$-clique enumeration algorithm~\cite{shi2020parallel}, which takes $\BigO{m\alpha^{s-2}}$ expected work and $\BigO{\log^2 n}$ span \whp{} Thus, in total, \ourapprox{} takes $\BigO{m\alpha^{s-2} }$ expected work and $\BigO{ \log^3 n}$ span \whp{}, as desired.

It remains to argue that \ourapproxnd{} computes an $(\binom{s}{r} + \varepsilon)$-approximate $(r, s)$-clique core number for every $r$-clique. 
Our proof here follows arguments by Ghaffari \textit{et al.}~\cite{Ghaffari2019}, in their approximate $k$-core decomposition algorithm. We generalize their arguments to apply to $(r, s)$ nucleus decomposition.

We prove this using induction on $i$, considering each bucket $B_i$. 
Our inductive hypothesis is that before beginning to peel $B_i$, we have already peeled all $r$-cliques with $(r, s)$-clique core numbers $\leq (1+\delta)^i$.
We show here that after a logarithmic number of rounds peeling all $r$-cliques in $B_i$, we have necessarily peeled all $r$-cliques with $(r, s)$-clique core numbers $\leq (1+\delta)^{i+1}$.
The key to this argument is Lemma~\ref{lemma:2}, where we set $\ell = (1+\delta)^{i+1}$. 
Notably, all $r$-cliques in $B_i$ have $s$-clique-degree at most $(\binom{s}{r} + \delta)\cdot (1 + \delta)^{i+1} = (\binom{s}{r} + \delta)\cdot \ell$, by construction of $B_i$.
Thus, Lemma~\ref{lemma:2} bounds the number of unpeeled $r$-cliques outside of the bucket $B_i$ (that is to say, in a bucket $B_j$ where $j > i$), but with $(r, s)$-clique core number $\leq \ell$ (these $r$-cliques are precisely those in $F_\ell$). 
Specifically, after each round of peeling $B_i$, at most $\slfrac{\binom{s}{r} \cdot |S_\ell|}{(\binom{s}{r} + \delta)}$ $r$-cliques remaining outside of $B_i$ have core number $\leq \ell$, so it takes $\BigO{\log_{1 + \delta / \binom{s}{r}}(n)}$ rounds until no $r$-cliques with core number $\leq \ell$ are outside of $B_i$ (that is to say, $F_\ell$ is empty). 
This means that after a logarithmic number of rounds of peeling all $r$-cliques in $B_i$, we have necessarily peeled all $r$-cliques with core number $\leq \ell$. 

Now, in our algorithm, we assign the approximate core number of these peeled $r$-cliques to be the upper boundary of $B_i$, or $(\binom{s}{r} + \delta)\cdot (1 + \delta)^{i+1}$. 
Note that the core numbers of the $r$-cliques that we peel while processing $B_i$ are $> (1+\delta)^i$ (by our inductive hypothesis) and $\leq \ell =  (1+\delta)^{i+1}$. Thus, our approximation is within a $(\binom{s}{r} + \delta)\cdot (1 + \delta) = (\binom{s}{r} + \varepsilon)$ factor of the true core number. 
Thus, \ourapproxnd{} gives a $(\binom{s}{r} + \varepsilon)$-approximation of the true core numbers of each $r$-clique.

\end{proof}

\else

\fi
\section{Practical Implementations}\label{sec:practical}

While the algorithm presented in Section~\ref{sec:alg} (Algorithm~\ref{alg:ndh}) is efficient in theory, we present a number of optimizations that improve its practical performance. Algorithm~\ref{alg:ndh} requires two passes over the $r$-cliques and their $s$-clique-adjacent neighbors, first to compute the $(r, s)$-clique core numbers and then to construct the hierarchy. 
We present algorithms that interleave these two computations, so that only a single pass is required. 
Specifically, we present two algorithms that are not as theoretically efficient as Algorithm~\ref{alg:ndh}, but are faster in practice, particularly when the difference between $r$ and $s$ is large, as we demonstrate in Section~\ref{sec:eval-compare}.

\subsection{Interleaved Hierarchy Framework}
Our algorithms use the same framework, given in Algorithm \ref{alg:ndh-framework}, and the main difference between the two algorithms is the implementation of the key subroutines \ourlink{} and \ourtree{}. 
The framework is based on the peeling process used to compute the $(r, s)$-clique core numbers in Shi \textit{et al.}'s work~\cite{shi2021nucleus}.
The main idea is that when we peel an $r$-clique $R$, while computing the updated $s$-clique counts due to peeling $R$, we are already iterating over all $s$-clique-adjacent $r$-cliques $R'$. Note that additionally, the nucleus decomposition algorithm in Shi \textit{et al.}~\cite{shi2021nucleus} uses a bucketing structure that maintains the intermediate $(r, s)$-clique core numbers of each $r$-clique throughout the peeling process (which begin as simply the $s$-clique count of each $r$-clique, and throughout the peeling process are updated to the actual $(r, s)$-clique core numbers). 
Then, if the intermediate $(r, s)$-clique core number of $R'$ is less than or equal to that of $R$, as maintained in the bucketing structure, the intermediate $(r, s)$-clique core numbers of $R$ and $R'$ 
are actually the final $(r, s)$-clique core numbers of $R$ and $R'$, respectively. Thus, each such $R$ and $R'$ pair are connected in the nucleus decomposition hierarchy up to the level given by $\min(ND[R], ND[R'])$, and after the peeling process completes, we will have all of the relevant connectivity information to completely compute the nucleus decomposition hierarchy. In this sense, it suffices to define a \ourlink{} subroutine to process the $s$-clique-adjacent $r$-cliques given the intermediate $(r, s)$-clique core numbers throughout the peeling algorithm. Based on \ourlink{}, \ourtree{} constructs the final hierarchy tree.

\begin{algorithm}[!t]
  \footnotesize
 \begin{algorithmic}[1]
 \State Initialize $r$, $s$ \Comment{$r$ and $s$ for $(r,s)$ nucleus decomposition}
\Procedure {\ourframework{}}{$G = (V,E)$}
\State $DG \leftarrow$ \algname{Arb-Orient}($G$) \Comment{Apply an arboricity-orientation algorithm}  \label{alg-fw:orient}
\State Initialize $U$ to be a parallel hash table with $r$-cliques as keys, and $s$-clique counts as values \label{alg-fw:init-table}
\State \ourcountrec{}($DG$, $s$, $U$) \Comment{Count $s$-cliques, and store the counts per $r$-clique in $U$} \label{alg-fw:count}
\State Let $ND$ be a bucketing structure mapping each $r$-clique to a bucket based on \# of $s$-cliques \label{alg-fw:init-bucket}
\State $\mathsf{finished} \leftarrow 0$
\While{$\mathsf{finished} < |U|$} \label{alg-fw:begin-peel}
\State $A \leftarrow$ $r$-cliques in the next bucket in $ND$ (to be peeled) \label{alg-fw:set-a}
\State $\mathsf{finished} \leftarrow \mathsf{finished} + |A|$
\ParFor{all $r$-cliques $R$ in $A$} \label{alg-fw:pair-start}
  \ParFor{all $s$-cliques $S$ containing $R$}
    \ParFor{all $r$-cliques $R'$ in $S$ where $R' \neq R$} \label{alg-fw:pair-end}
      \If{$ND[R'] \leq ND[R]$} \label{alg-fw:if}
        \State \ourlink{}$(R', R, ND)$ \label{alg-fw:link}
      \Else
        \ Update $s$-clique count of $R'$ in $U$ \label{alg-fw:update}
      \EndIf
    \EndParFor
  \EndParFor
\EndParFor 
\State Update the buckets of $r$-cliques with updated $s$-clique counts, peeling $A$ \label{alg-fw:update-bucket}
\EndWhile \label{alg-fw:end-peel}
\State \Return \ourtree{}($ND$) \Comment{Return the hierarchy tree $T$, constructed based on \ourlink{}} \label{alg-fw:tree}
  \EndProcedure
 \end{algorithmic}
\caption{Parallel $(r,s)$ nucleus hierarchy framework}
 \label{alg:ndh-framework}
\end{algorithm}

In more detail, \ourframework{} first uses an efficient low out-degree orientation algorithm by Shi \textit{et al.}~\cite{shi2020parallel} to direct the graph $G$ such that every vertex has out-degree at most $O(\alpha)$ (Line~\ref{alg-fw:orient}). Then, it counts the number of $s$-cliques per $r$-clique in $G$ and stores the counts in a parallel hash table $U$, where the keys are $r$-cliques and the values are the counts (Lines~\ref{alg-fw:init-table}--\ref{alg-fw:count}). Note that \ourframework{} uses a subroutine \ourcountrec{} based on previous work by Shi \textit{et al.}~\cite{shi2020parallel}, to count the number of $s$-cliques per $r$-clique. Also, our algorithm initializes a parallel bucketing structure $ND$ that maps $r$-cliques to buckets, initially based on their $s$-clique counts (Line~\ref{alg-fw:init-bucket}). 
We use the bucketing structure by Dhulipala et al.~\cite{DhBlSh17}. This structure $ND$ stores the aforementioned intermediate $(r, s)$-clique core numbers of each $r$-clique, and supports efficient operations to update buckets and return the lowest unpeeled bucket.
Our algorithm then proceeds with a classic peeling paradigm, where until all $r$-cliques have been peeled, it processes the $r$-cliques (that have not yet been peeled) incident to the lowest number of $s$-cliques and peels them from the graph (Lines~\ref{alg-fw:begin-peel}--\ref{alg-fw:end-peel}). For a set $A$ of $r$-cliques with the lowest number of incident $s$-cliques (Line~\ref{alg-fw:set-a}), we iterate over all $s$-clique-adjacent $r$-cliques $R'$ to each $r$-clique $R$ in $A$ (Lines~\ref{alg-fw:pair-start}--\ref{alg-fw:pair-end}). 

Note that if $ND[R'] \leq ND[R]$, this means that $R'$ was either previously peeled or is currently being peeled (and is also in $A$); this is because at any given peeling step, we process the bucket of unpeeled $r$-cliques with the minimum incident $s$-clique count. This also means that $ND[R']$ and $ND[R]$ are the actual $(r,s)$-clique core numbers of $R'$ and $R$ respectively, which follows directly from the correctness of the peeling paradigm~\cite{Sariyuce2017}. Thus, each such $R$ and $R'$ pair are connected in the nucleus decomposition hierarchy up to the level given by $\min(ND[R], ND[R'])$, which we process using the \ourlink{} subroutine (Lines~\ref{alg-fw:if}--\ref{alg-fw:link}). The \ourlink{} subroutine will construct the hierarchy, and we describe it in Sections~\ref{sec:basic-link} and~\ref{sec:efficient-link}.

On the other hand, if $ND[R] > ND[R']$, then this means that $R'$ has not yet been peeled, and the $s$-clique removed by $R$ must be properly accounted for. In this case, \ourframework{} updates the $s$-clique count of $R'$ in the hash table $U$ (Line~\ref{alg-fw:update}). After processing all $R$ and $R'$ pairs, we then update the buckets of the $r$-cliques with updated $s$-clique counts in $U$ (Line~\ref{alg-fw:update-bucket}). We omit the details of these steps for conciseness, since they are described in Shi \textit{et al.}'s parallel nucleus decomposition algorithm~\cite{shi2021nucleus}.

\begin{algorithm}[!t]
  \footnotesize
 \begin{algorithmic}[1]
 \State Initialize $k$ union-find data structures, $uf_i$ for $i \in [k]$, where $k$ is the maximum $(r, s)$-clique core number
\Procedure {\ourilink{}}{$R$, $Q$, $ND$}
\ParFor{$i \in [\min(ND[R], ND[Q])]$} \label{alg-ilink:first}
  \State $uf_i.\mathsf{unite}(R, Q)$ \label{alg-ilink:last}
\EndParFor
  \EndProcedure

\smallskip
\Procedure{\ouritree}{$ND$}
\State Initialize the hierarchy tree $T$ with leaves corresponding to each $r$-clique
\For{$i \in \{k, k-1, \ldots, 1\}$} \label{alg-itree:iterate}
  \ParFor{each connected component $C = \{R_1, \ldots, R_c\}$ in $uf_i$} \label{alg-itree:component}
    \State Construct a new parent in $T$, to be the parent of the roots of the leaf nodes corresponding to each $R_\ell$ (for $\ell \in [c]$) \label{alg-itree:construct}
  \EndParFor
\EndFor
\State \Return $T$
\EndProcedure
 \end{algorithmic}
\caption{Basic link and tree construction }
 \label{alg:ilink}
\end{algorithm}

\subsection{Basic Version of \ourlink{}} \label{sec:basic-link}
The remaining details are in how we perform the \ourlink{} subroutine (Line~\ref{alg-fw:link}) and how we construct the hierarchy tree $T$ with \ourtree{} (Line~\ref{alg-fw:tree}).

In Algorithm \ref{alg:ilink}, we present a basic \ourlink{} subroutine, \ourilink{}, and the corresponding \ourtree{} subroutine, \ouritree{}. 
\ourilink{} maintains a parallel union-find data structure $uf_i$ per core number $i \in [k]$, which corresponds to a level of the hierarchy tree $T$. Each $uf_i$ connects $r$-cliques that are $s$-clique-adjacent considering only $r$-cliques with core numbers $\geq i$. To construct these $uf_i$'s, given two $r$-cliques $R$ and $Q$, \ourilink{} simply unites $R$ and $Q$ in each $uf_i$ where $i \leq \max(ND[R], ND[Q])$ (Lines~\ref{alg-ilink:first}--\ref{alg-ilink:last}).
Then, given the $uf_i$ for all $i \in [k]$, we construct the hierarchy tree $T$ from the bottom up, starting with leaf nodes corresponding to $r$-cliques. 
\ouritree{} begins with $i = k$, where for each connected component in $uf_k$ (Line~\ref{alg-itree:component}), we construct a parent in $T$ where its children are the leaf nodes corresponding to the $r$-cliques in the connected component (Line~\ref{alg-itree:construct}). Then, for $i = k-1, \ldots, 1$ (Line~\ref{alg-itree:iterate}), we construct a new parent for each connected component in $uf_i$, where its children are the parents of the leaf nodes corresponding to the $r$-cliques that compose the component (Line~\ref{alg-itree:construct}). This produces the desired $T$. 

However, \ourilink{} is not efficient, since it requires a union-find data structure per level, and for every pair of $r$-cliques, we could perform up to $k$ \algname{unite} operations. Indeed, in Section~\ref{sec:eval-compare}, we empirically show that \ourilink{} performs many unnecessary \algname{unite} operations in practice. 
If we let $n_r$ and $n_s$ denote the number of $r$-cliques and the number of $s$-cliques in the graph, respectively, \ourilink{} incurs additional space proportional to $\BigO{k n_r}$ and total work upper bounded  by $\BigO{k n_s}$ (since there are at most $\BigO{n_s}$ pairs of $s$-clique-adjacent $r$-cliques).
In the next subsection, we introduce more efficient \ourlink{} and \ourtree{} subroutines.

\begin{algorithm}[!t]
  \footnotesize
 \begin{algorithmic}[1]
 \State Initialize a union-find data structure, $uf$, of length equal to the number of $r$-cliques
 \State Initialize a hash table, $L$, where the keys and values are $r$-cliques
\Procedure {\ourelink{}}{$R$, $Q$, $ND$}
\If{$R$ or $Q$ is empty} \Return \EndIf \label{alg-elink:empty}
\If{$ND[Q] < ND[R]$} Swap $R$ and $Q$ \EndIf \label{alg-elink:check}
\State $R \leftarrow uf.\mathsf{parent}(R)$, $Q \leftarrow uf.\mathsf{parent}(Q)$ \label{alg-elink:set-parent}
\If{$ND[R] = ND[Q]$} \label{alg-elink:begin-if}
  \State $uf.\mathsf{unite}(R, Q)$ \label{alg-elink:if-unite}
  \If{$uf.\mathsf{parent}(R) \neq R$} \ourelink{}$(L[R], uf.\mathsf{parent}(R), ND)$ \EndIf \label{alg-elink:if-1}
  \If{$uf.\mathsf{parent}(Q) \neq Q$} \ourelink{}$(L[Q], uf.\mathsf{parent}(Q), ND)$ \EndIf \label{alg-elink:if-2}
\Else \Comment{$ND[R] < ND[Q]$}\label{alg-elink:begin-else}
  \While{true} \label{alg-elink:while-true}
  \State $LQ \leftarrow L[Q]$
  \State $Q \leftarrow uf.\mathsf{parent}(Q)$ \label{alg-elink:get-parent}
    \If{\cas{}($L[Q]$, empty, $R$)} \label{alg-elink:else-empty}
      \If{$uf.\mathsf{parent}(Q) \neq Q$} \label{alg-elink:empty-parent-start}
        \State \ourelink{}($R, uf.\mathsf{parent}(Q), ND$) 
      \EndIf \label{alg-elink:empty-parent-end}
      \State break
    \ElsIf{$ND[LQ] < ND[R]$} \label{alg-elink:else-check-nd}
      \If{\cas{}($L[Q]$, $LQ$, $R$)} \label{alg-elink:else-nd-cas}
        \If{$uf.\mathsf{parent}(Q) \neq Q$} \label{alg-elink:nd-parent-start}
          \State \ourelink{}($R, uf.\mathsf{parent}(Q), ND$)
        \EndIf \label{alg-elink:nd-parent-end}
        \State \ourelink{}($R, LQ, ND$) \label{alg-elink:nd-r-lq}
        \State break \label{alg-elink:break2}
      \EndIf
    \Else \label{alg-elink:actual-else}
      \State \ourelink{}($R, L[Q], ND$) \label{alg-elink:act-actual-else}
      \State break \label{alg-elink:end}
    \EndIf
  \EndWhile
\EndIf
  \EndProcedure

\smallskip
\Procedure{\ouretree}{$ND$}
\State Initialize the hierarchy tree $T$ with leaves corresponding to each $r$-clique \label{alg-etree:init-tree}
\ParFor{each connected component $\mathcal{C} = \{R_1, \ldots, R_c\}$ in $uf$} \label{alg-etree:uf-begin-parent}
  \State Construct a parent node $uf_\mathcal{C}$ in $T$, where its children are the leaves corresponding to the $r$-cliques in $\mathcal{C}$ \label{alg-etree:uf-parent}
\EndParFor \label{alg-etree:uf-end-parent}
\ParFor{each parent node $uf_\mathcal{C}$ in $T$} \label{alg-etree:uf-begin-link}
  \If{$L[\mathcal{C}]$ is non-empty} \label{alg-etree:if}
    \State $R \leftarrow uf.\mathsf{parent}(L[\mathcal{C}])$ \Comment{Note that $\mathcal{C}$ is the $r$-clique representing the component in $uf$} \label{alg-etree:uf-r}
    \State Make $uf_\mathcal{C}$ a child of $uf_R$ in $T$ \Comment{Note that $uf_R$ necessarily exists since $R$ represents a component in $uf$} \label{alg-etree:uf-link}
  \EndIf
\EndParFor \label{alg-etree:uf-end-link}
\State \Return $T$
\EndProcedure
 \end{algorithmic}
\caption{Efficient link and tree construction }
 \label{alg:elink}
\end{algorithm}

\subsection{Efficient Version of \ourlink{}}  \label{sec:efficient-link}
Our improved subroutines \ourelink{} and \ouretree{} are shown in Algorithm \ref{alg:elink}. 
We refer to an example of $(1, 3)$-nucleus decomposition in the graph in Figure~\ref{fig:graph}. Recall that we have omitted labeling some vertices in the graph for simplicity.

The main idea of \ourelink{} is instead of maintaining $k$ union-find data structures, we maintain a single parallel union-find data structure $uf$ and an additional hash table $L$ that maps $r$-cliques to $r$-cliques. 
First, $uf$ stores connected $r$-cliques considering only other $r$-cliques with equal core numbers. 
For instance, in Figure~\ref{fig:graph}, the vertices $3a$, $3b$, and $3c$ are connected and all have core number 3, so we would store these as a component in $uf$. Note that for all core numbers $i$, we can store this information using a single union-find data structure because the sets of $r$-cliques with distinct core numbers are disjoint. 
We can arbitrarily represent each connected component in $uf$ by a single $r$-clique in that component.

The main idea of $L$ is to connect the components in $uf$ to the ``nearest'' core with a different core number that it is contained within (if it exists). 
For instance, in Figure~\ref{fig:graph}, we note that the component in $uf$ corresponding to $4a$ (consisting of vertices with core number 4) is contained within the 3-core consisting of the component $\{3a, 3b, 3c\}$. In $L$, we would store one of $3a$, $3b$, or $3c$ in an entry corresponding to key $4a$, indicating that this is the ``nearest'' core that $4a$ must join in the hierarchy.
We note that $4a$ is also contained within a larger 2-core and a larger 1-core, but its ``nearest'' core, or the smallest core such that the component $4a$ is a proper subset of that core, is given by the 3-core. 
It is sufficient to store only the ``nearest'' core to $4a$ in $L$, because the component in $uf$ corresponding to $\{3a, 3b, 3c\}$ is responsible for storing its ``nearest'' core in $L$ as well, to the 2-core it is contained within. 
On the other hand, the component corresponding to $4d$ is not contained within the 3-core, and its ``nearest'' core would be the 2-core containing the component $2a$, so $4d$ would store in $L$ the value $2a$.

More formally, for each $r$-clique $R$ representing a connected component in $uf$, let $R'$ be an $r$-clique with the maximum $ND[R']$ such that $ND[R'] < ND[R]$, and $R'$ is connected to $R$ through $s$-cliques considering only $r$-cliques with core number $\geq ND[R']$.
Then, each such $r$-clique $R$ is a key in $L$, and $L$ stores the corresponding $R'$ that satisfies these conditions as the value.
Note that if there are multiple such $R'$ where $ND[R']$ is maximized under these conditions, it is irrelevant which $R'$ is stored in $L$, because it is simple to look up the component that $R'$ corresponds to using $uf$. That is to say, it is irrelevant which of $3a$, $3b$, and $3c$ we store for $4a$ in $L$, because we can look up the parent of $3a$, $3b$, and $3c$ in $uf$, which would resolve to the same parent.
\ourelink{} updates $uf$ and $L$, depending on the core numbers of the given $r$-cliques $R$ and $Q$.

\begin{figure*}[t]
\vspace{-3pt}
   \begin{minipage}{0.5\textwidth}
 \centering
   \includegraphics[width=0.8\textwidth]{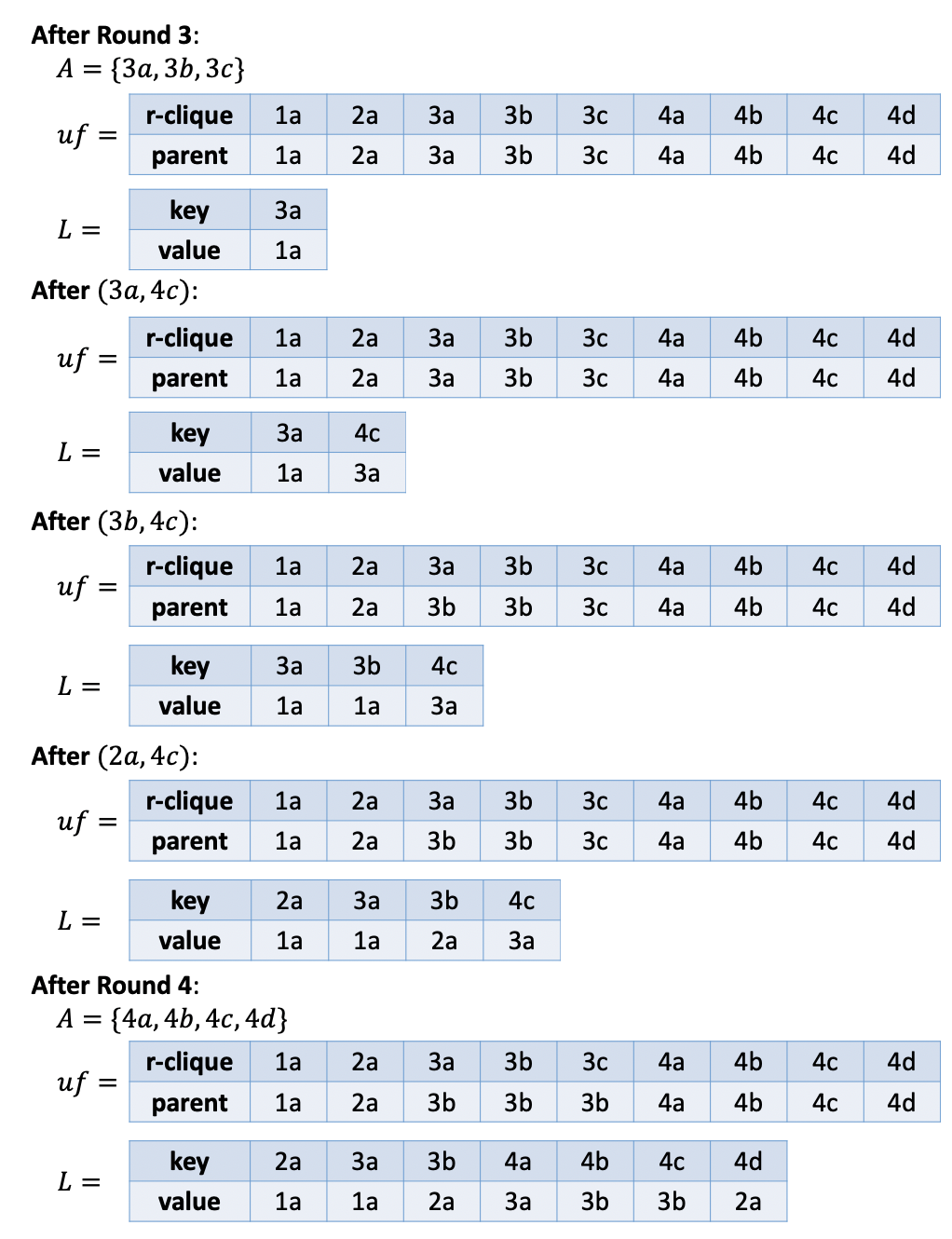}
   \vspace{-5pt}
   \caption{
    An example of the $uf$ and $L$ data structures maintained by \ourelink{} when computing the $(1,3)$-nucleus hierarchy on the graph in Figure~\ref{fig:graph}.
   The data structures are shown after the third and fourth rounds of peeling in \ourframework{}, and after intermediate calls to \ourelink{} within the fourth round. 
   }
    \label{fig:e-link}
   
   \end{minipage}\hfill
   \begin{minipage}{0.47\textwidth}
     \centering

   \includegraphics[width=0.6\textwidth]{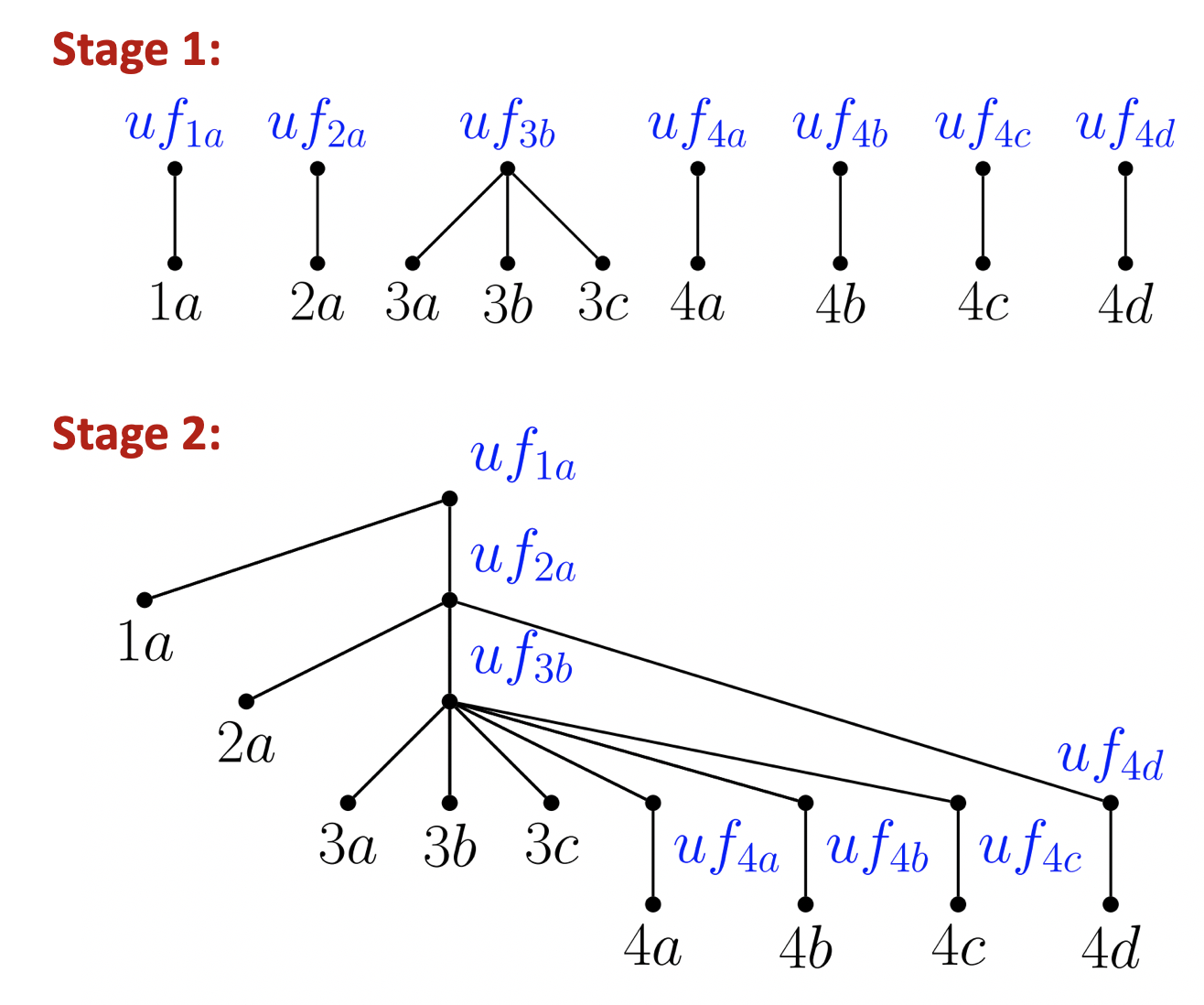}
   \vspace{-5pt}
   \caption{
   An example of the $(1,3)$-nucleus hierarchy tree on the graph in Figure~\ref{fig:graph}, constructed by \ouretree{}. Stage 1 shows an intermediate tree constructed in the \ouretree{} subroutine, and Stage 2 depicts the final hierarchy tree.
   }
    \label{fig:e-link-tree}
   \end{minipage}
\end{figure*}

\myparagraph{Tree Construction} We describe first how \ouretree{} constructs the hierarchy tree $T$ given the specifications for $uf$ and $L$. We refer to the construction shown in Figure~\ref{fig:e-link-tree}, where $uf$ and $L$ are given under ``After Round 4'' in Figure~\ref{fig:e-link}.

The main idea for the construction is that if we begin with a hierarchy tree $T$ consisting of only leaf nodes corresponding to each $r$-clique, the highest (bottom-most) level in which that leaf node may join a non-trivial connected component is the level corresponding to the leaf node's $r$-clique's $(r, s)$-clique core number. That is to say, starting from $i = k$ and iterating to $i = 1$, an $r$-clique $R$ is necessarily a singleton component from $i = k$ to $i = ND[R] + 1$.
The union-find structure $uf$ denotes the parent of each leaf node on the level corresponding to its core number.

Then, once we have the connected components per level corresponding to each core number, it remains to describe how components with different core numbers are contained within each other. This containment is necessarily hierarchical, where components corresponding to larger core numbers are contained within successive components corresponding to smaller core numbers, as can be seen in Figure~\ref{fig:graph}. 
This containment is precisely what $L$ describes; it points each component within a given core $i$ to a component on the greatest core number $j$ such that $j < i$, where the two components are connected on level $j$ (that is to say, connected through $r$-cliques with core numbers $\geq j$). 
Thus, using $L$, we can point each parent constructed using our $uf$ to another parent, which represents the first instance (from the bottom up in $T$) in which the original parent is merged into another component.

In more detail, in \ouretree{}, we begin with a hierarchy tree $T$ consisting of isolated $r$-cliques (Line~\ref{alg-etree:init-tree}).
For each component $\mathcal{C} = \{R_1, \ldots, R_c\}$ in $uf$, representing a connected component of $r$-cliques with the same core number, we construct a parent in $T$, where its children are leaves corresponding to the $r$-cliques in $\mathcal{C}$ (Lines~\ref{alg-etree:uf-begin-parent}--\ref{alg-etree:uf-end-parent}). 
$\mathcal{C}$ is one of these cliques, and is chosen arbitrarily. 
The new parent is $uf_\mathcal{C}$.
This process is shown in Stage 1 in Figure~\ref{fig:e-link-tree}; each vertex has itself as its parent in $uf$, except the vertices $3a$, $3b$, and $3c$, which have $3b$ as their parent. Thus, we create a new parent node $uf_{3b}$ for the leaves $3a$, $3b$, and $3c$.

Then, for each parent node $uf_\mathcal{C}$, we look up the nearest connected component that it should hierarchically connect to using $L$. What this means is that the component $uf_\mathcal{C}$ should join the connected component of $L[\mathcal{C}]$ (if it exists). If $R=uf.\mathsf{parent}(L[\mathcal{C}])$, then the connected component of $L[\mathcal{C}]$ is represented by $uf_R$.
Thus, we make $uf_\mathcal{C}$ a child of $uf_R$ in $T$ (Lines~\ref{alg-etree:uf-begin-link}--\ref{alg-etree:uf-end-link}).
This construction is shown in Stage 2 in our example in Figure~\ref{fig:e-link-tree}. We note that $L[2a] = 1a$ and $uf[1a] = 1a$, so we make $uf_{2a}$ the child of $uf_{1a}$. Similarly, $L[3b] = 2a$ and $uf[2a] = 2a$, so we make $uf_{3b}$ the child of $uf_{2a}$, and $L[4d] = 2a$ and $uf[2a] = 2a$, so we also make $uf_{4d}$ the child of $uf_{2a}$. Finally, we have $L[4a] = 3a$, $L[4b] = 3b$, and $L[4c] = 3b$, where $uf[3a] = uf[3b] = 3b$, so we make $uf_{4a}$, $uf_{4b}$, and $uf_{4c}$ the children of $uf_{3b}$.
In this manner, we construct $T$. 

We make a subtle note that it is not strictly necessary to maintain parent nodes in the hierarchy tree with exactly one child; we can remove all such parents $P$ and make the child $C$ a direct child of its grandparent $G$. 
This is because it is inherently implied that the component that the child $C$ represents in unchanged until it joins the component represented by $G$. 
For instance, $uf_{4d}$'s sole child is $4d$. If we set $4d$ to be a direct child of its grandparent, $uf_{2a}$, and remove $uf_{4d}$, then it is implied that the vertex $4d$ remains its own component until it joins the 2-core containing $2a$ (and notably, $4d$ represents an unchanged component in the 3-core and in the 4-core of the graph).
In this sense, the final hierarchy tree produced in Figure~\ref{fig:e-link-tree} is actually equivalent to that in Figure~\ref{fig:alg1-tree}.

\myparagraph{\ourlink{} Subroutine} 
We now describe \ourelink{}.
The key to \ourelink{} is to properly maintain $uf$ and $L$ according to their definitions given new information obtained by connected $r$-cliques. The main subtlety is that after updating the connectivity information of $R$ or $Q$, in $uf$ or in $L$, there may be cascading effects resulting in new calls to \ourelink{}.

In more detail, \ourelink{} first ensures that $ND[R] \leq ND[Q]$ (Line~\ref{alg-elink:check}). We need only perform operations on the parents of the components of the $r$-cliques in $uf$, so we set $R$ and $Q$ to their respective parents in $uf$ (Line~\ref{alg-elink:set-parent}).

The first case is if $ND[R] = ND[Q]$ (Line~\ref{alg-elink:begin-if}). Then, we need only unite $R$ and $Q$ in $uf$, to maintain that $uf$ tracks the connectivity between $r$-cliques with the same core number (Line~\ref{alg-elink:if-unite}). 
However, the new parent $P$ may now need to update its value in $L$ based on $L[R]$ and $L[Q]$.  
This is because it is possible that $L[P]$ must be updated; 
for instance, if $L[R]$ has a strictly greater core number than the current $L[P]$, $L[P]$ must be updated to be $L[R]$. 
\ourelink{} performs these updates by calling itself on $L[R]$ and $P$, and on $L[Q]$ and $P$, if $P\neq R$ and $P\neq Q$, respectively (Lines~\ref{alg-elink:if-1}--\ref{alg-elink:if-2}).

The second case is if $ND[R] < ND[Q]$ (Line~\ref{alg-elink:begin-else}). There are two main considerations. First, $R$ may replace the current value of $L[Q]$, which we call $LQ$, if $ND[R] > ND[LQ]$. Second, $R$ and $LQ$ must be linked, since they are connected through $Q$ and could affect each other's parent or value in $uf$ or $L$, respectively. We handle these cases with a series of if statements and \cas{}s.

We first perform a \cas{}, checking if $L[Q]$ is empty and replacing it with $R$ if so (Line~\ref{alg-elink:else-empty}); if this \cas{} succeeds, then there is still the possibility that $Q$'s parent in $uf$ changed before the \cas{} completed, in which case $Q$'s new parent is unaware of its connection to $R$. In this scenario, we must call \ourelink{} again on $R$ and the new parent of $Q$, since $R$ could potentially modify the $r$-clique stored in $L$ corresponding to $Q$'s new parent (Lines~\ref{alg-elink:empty-parent-start}--\ref{alg-elink:empty-parent-end}).

If the previous \cas{} failed, then we check if $ND[LQ] < ND[R]$ (Line~\ref{alg-elink:else-check-nd}), which if true, means that $R$ is a candidate to replace $LQ$ in $L$. 
We perform another \cas{} to replace $LQ$ with $R$ (Line~\ref{alg-elink:else-nd-cas}), and if it succeeds, we must again check if $Q$'s parent in $uf$ has potentially changed before the \cas{} completed. Again, if this occurs, we must call \ourelink{} on $R$ and the new parent $Q$ (Lines~\ref{alg-elink:nd-parent-start}--\ref{alg-elink:nd-parent-end}). 
We must also call \ourelink{} on $R$ and $LQ$ (Line~\ref{alg-elink:nd-r-lq}), to store $R$'s connectivity to $LQ$. This is because $R$'s ``nearest'' core as stored in $L$ may be superseded by $LQ$. 
If the \cas{} fails, we simply try again, hence the while loop (Line~\ref{alg-elink:while-true}).

The last case is if $ND[LQ] \geq ND[R]$ (Line~\ref{alg-elink:actual-else}), in which case $R$ is not a candidate to replace $LQ$ in $L$, and we store $R$'s connectivity to $Q$ by calling \ourelink{} on $R$ and $L[Q]$. Note that this is necessary because $R$ and $L[Q]$ could be united in $uf$ if they have the same core number, or $R$ could be a ``nearest'' core to $L[Q]$.

This concludes our efficient \ourlink{} subroutine, \ourelink{}. We show in Section~\ref{sec:eval-compare} that \ourelink{} performs many fewer \algname{unite} and \ourlink{} operations, and achieves significant speedups, over \ourilink{}.
We provide an example of running \ourelink{}.

\myparagraph{Example of \ourelink{}}
As an example of these cascading effects, we refer to an intermediate state of $uf$ and $L$ on the example graph in Figure~\ref{fig:graph}, given after the third round of peeling in \ourframework{} (immediately before peeling the final set of vertices in the graph, given by the components $4a$, $4b$, $4c$, and $4d$). 
This state is shown in the data structures under ``After Round 3'' in Figure~\ref{fig:e-link}. In $uf$, every vertex is its own parent, and in $L$, we have identified that the component corresponding to $3a$ is connected to $1a$. 
Note that many link operations occur from peeling $4a$, $4b$, $4c$, and $4d$, including for the $(R, Q)$ pairs $(3a, 4c)$, $(3b, 4c)$, and $(2a, 4c)$, which we consider in this example. 
We show in Figure~\ref{fig:e-link} the state of $uf$ and $L$ after each of these three calls to \ourelink{} (including the cascading calls that these calls generate).
We list the $R$ and $Q$ for each \ourelink{} operation, as well as the cascading calls that they invoke. We assume the operations happen sequentially for clarity, although in practice they can happen concurrently.

\begin{itemize}[leftmargin=*]
    \item $R = 3a$, $Q = 4c$: We now know that $4c$ is connected to $3a$'s component, and $4c$ does not have a previously set ``nearest'' core, so we set $L[4c] = 3a$ (Line~\ref{alg-elink:else-empty}).
    \item $R = 3b$, $Q = 4c$: We note that $L[4c]$ is now already set to a ``nearest'' core with core number 3, so there is nothing new to update for $L[Q]$. However, we gain from this new link the knowledge that $3a$ and $3b$ are connected (since $L[4c]=3a$), so we must cascade a new \ourelink{} call to $(R,L[Q]) = (3b, 3a)$ (Line~\ref{alg-elink:act-actual-else}), so that $3a$ and $3b$ can be set to the same component in $uf$.
    \begin{itemize}
        \item $R = 3a$, $Q = 3b$: We now call \algname{unite} on $3a$ and $3b$ in $uf$ (Line~\ref{alg-elink:if-unite}). Say that arbitrarily, $3b$ is set as the new parent of $3a$ and $3b$ in $uf$. Since the parents in $uf$ are responsible for maintaining the connection to the ``nearest'' core, we must transfer $L[R] = L[3a]$ to $L[Q] = L[3b]$. We do so by calling \ourelink{} on $(L[R], uf.\mathsf{parent}(R)) = (1a, 3b)$ (Line~\ref{alg-elink:if-1}).
        \begin{itemize}
            \item $R = 1a$, $Q = 3b$: Now, we can set $L[Q] = L[3b]$ to $R = 1a$, since $3b$'s ``nearest'' core is now $1a$ (Line~\ref{alg-elink:else-empty}).
        \end{itemize}
    \end{itemize}
    \item $R = 2a$, $Q = 4c$: Since $4c$ already has an entry in $L$ that's ``nearer'' to it than $2a$, there is nothing new to update for $L[Q]$. However, we now know that $3a$ and $2a$ are connected, so we must cascade a new \ourelink{} call to $(R, L[Q]) = (2a, 3a)$ (Line~\ref{alg-elink:act-actual-else}).
    \begin{itemize}
        \item $R = 2a$, $Q = 3a$: Since the parent of $Q$ in $uf$ is $3b$, we can treat this as $R = 2a$ and $Q = 3b$ (since we only need to maintain connections in $L$ for the parents in $uf$). Now, we find that $3b$ has recorded its ``nearest'' core as $L[3b] = 1a$, but $2a$ is ``nearer''. Thus, we update $L[3b]$ to be $2a$ (Line ~\ref{alg-elink:else-nd-cas}), but now we know that $2a$ is connected to $1a$. So, we call \ourelink{} on $(R, L[Q]) = (2a, 1a)$ (Line~\ref{alg-elink:nd-r-lq}).
        \begin{itemize}
            \item $R = 1a$, $Q = 2a$: We discover that $2a$'s ``nearest'' core is given by $1a$. We set $L[Q] = L[2a]$ to $R = 1a$ (Line~\ref{alg-elink:else-empty}).
        \end{itemize}
    \end{itemize}
\end{itemize}

Note that it is necessary for us to perform these cascading calls to \ourelink{}, because the only way for $3a$ and $3b$ to discover that they should be connected is through one of the 4-core components, and the only way for $3b$ to realize that the component with $2a$ is its ``nearest'' core is also through one of the 4-core components. Similarly, the only way for $2a$ to realize that the component with $1a$ is its ``nearest'' core is through first one of the 4-core components, then through the 3-core component, in which $3a$, which we now know is connected to $3b$, had the original adjacency to $1a$. Thus, information must be constantly propagated through $uf$ and $L$.

\myparagraph{Comparison to Prior Work}
\sariyuce{} and Pinar~\cite{SaPi16} also provide a hierarchy construction algorithm, \algname{nh}, that is performed interleaved with the peeling process. 
They maintain a union-find data structure that stores the connectivity of all $r$-cliques considering only $r$-cliques with the same core number. 
However, for adjacent $r$-cliques with different core numbers, they simply store all pairs of such $r$-cliques in a list. 
They process this list after the peeling process to construct the hierarchy tree, and their method for processing this list requires a global view, since \algname{nh} first sorts the pairs of $r$-cliques in the list based on their core numbers.
Storing this list incurs additional space potentially proportional to the number of $s$-cliques in the graph, which is a significant overhead.

Our main innovation in \ourelink{} is that we need only incur additional space overhead proportional to the number of $r$-cliques in the graph, because we process adjacent $r$-cliques with different core numbers while performing the peeling process. We are able to process this information into a hash table $L$ proportional to the number of $r$-cliques, so our memory overhead overall is $2 n_r$, where $n_r$ is the number of $r$-cliques. In contrast, \algname{nh} uses $\binom{s}{r} \cdot n_s + n_r$ additional space, where $n_s$ is the number of $s$-cliques.
In addition, \algname{nh} is sequential, whereas \ourelink{} is thread-safe and carefully resolves conflicts in updating $uf$ and $L$. Also, the post-processing step to construct the hierarchy tree in \algname{nh} involves many sequential dependencies, where even merges on the same level of the tree may conflict with each other, whereas our post-processing step, \ouretree{}, is fully parallel.

\subsection{Practical Version of \ourndh{}} \label{sec:ndh-practice}
Finally, we make certain optimizations to our theoretically-efficient $(r, s)$ nucleus decomposition hierarchy algorithm, \ourndh{} (Algorithm~\ref{alg:ndh}) to improve its performance in practice. We maintain the two-pass paradigm of first computing the $(r, s)$-clique core numbers of each $r$-clique and then constructing the hierarchy tree $T$. 
However, we do not explicitly store linked lists containing all pairs of $s$-clique-adjacent $r$-cliques, since this represents too much of a memory overhead to be practical, particularly for larger $r$ and $s$. We also do not explicitly generate the graph $H$ given by these linked lists (Line~\ref{alg-ndh:list-rank}).

Instead, we use a single union-find data structure to maintain the connected components (throughout the loop on Lines~\ref{alg-ndh:start-big-loop}--\ref{alg-ndh:end-big-loop}), and for each $i \in \{k, k-1, \ldots, 1\}$ (Line~\ref{alg-ndh:start-big-loop}), we iterate through all $r$-cliques $R$ with core number $i$ and their $s$-clique-adjacent $r$-cliques $R'$. 
We perform a parallel sort on the $r$-cliques based on their core numbers, which allows us to efficiently extract $r$-cliques with the same core numbers; this adds a small additional memory overhead, which we observe in Section~\ref{sec:eval-compare}.
For each such pair of $r$-cliques $R$ and $R'$ where $ND[R'] \geq ND[R]$, we directly unite them in our union-find data structure to obtain the desired connected components; this replicates the same information stored in the linked lists (on Lines~\ref{alg-ndh:start-init-hash}--\ref{alg-ndh:end-init-hash}). 
We construct the requisite new parents in the hierarchy tree $T$ given the computed connected components, and we reuse the same union-find data structure for subsequent $i$.

\section{Evaluation} \label{sec:eval}

\begin{table}[t]
\footnotesize
    \centering
\begin{tabular}{lll}
& $n$         & $m$                  \\ \midrule
\textbf{amazon} &334,863   & 925,872  \\ \hline
\textbf{dblp} &317,080   & 1,049,866  \\ \hline
\textbf{youtube} & 1,134,890  & 2,987,624  \\ \hline
\textbf{skitter} &1,696,415 & 11,095,298  \\ \hline
\textbf{livejournal} &3,997,962 & 34,681,189  \\ \hline
\textbf{orkut} & 3,072,441 & 117,185,083 \\ \hline
\textbf{friendster} & 65,608,366 & $1.806\times 10^9$  \\ \hline
\end{tabular}
\caption{Sizes of our input graphs, which are from SNAP~\cite{SNAP}.}
\vspace{-5pt}
\label{table:graphs}
\end{table}

\myparagraph{Environment and Inputs}
We run our experiments on a Google Cloud Platform instance with a 30-core machine with two-way hyper-threading, with 3.9 GHz Intel Cascade Lake processors and 240 GB of main memory. We use all cores when testing parallel implementations, unless specified otherwise. 
Our implementations are written in C++ and
we compile our code using g++ (version 7.4.0) with the \texttt{-O3} flag.
We use parallel primitives and the work-stealing scheduler from \algname{ParlayLib} by Blelloch \textit{et al.}~\cite{BlAnDh20}. We terminate any experiment that takes over 4 hours.
We test our algorithms on real-world graphs from the Stanford Network Analysis Project (SNAP)~\cite{SNAP}, shown in 
Table~\ref{table:graphs}.

We implement all three versions of our exact $(r, s)$ nucleus decomposition hierarchy algorithms, including our theoretically-efficient \ourndh{} (Algorithm~\ref{alg:ndh}, 
with the optimizations described in Section~\ref{sec:ndh-practice}), 
which we call \ourte{}, and our nucleus decomposition hierarchy framework \ourframework{} (Algorithm~\ref{alg:ndh-framework}) using both \ourilink{} (Algorithm~\ref{alg:ilink}) and \ourelink{} (Algorithm~\ref{alg:elink}), which we call \ourbl{} and \ourel{}, respectively.

We also implement our approximate $(r, s)$ nucleus decomposition hierarchy algorithm, \ourapprox{} (using Algorithm~\ref{alg:ndh-approx}). We integrate \ourapproxnd{} with each of \ourte{}, \ourel{}, and \ourbl{} for the hierarchy construction, giving us three implementations, \ourapproxte{}, \ourapproxel{}, and \ourapproxbl{}, respectively.

We compare our hierarchy algorithms to \algname{nh}, the state-of-the-art sequential $(r, s)$ nucleus decomposition hierarchy implementation by \sariyuce{} and Pinar~\cite{SaPi16} for $(1, 2)$, $(2, 3)$, and $(3, 4)$ nucleus decomposition. Note that \algname{nh} does not generalize to other $r$ and $s$ values. \algname{nh}, like \ourbl{} and \ourel{}, constructs the hierarchy while computing the $(r,s)$-clique-core numbers of the graph. 
For the special case of $k$-core, we also compare to \algname{phcd}, the state-of-the-art parallel $k$-core hierarchy implementation by Chu \textit{et al.}~\cite{Chu2022seqkcore}.

\begin{figure*}[t]
    \centering
    \vspace{-5pt}
   \begin{subfigure}{.35\textwidth}
   \centering 
   \includegraphics[width=\columnwidth, page=5]{figures/fig_more.pdf}
   \end{subfigure}%
   \hfill
   \begin{subfigure}{.35\textwidth}
   \centering %
   \includegraphics[width=\columnwidth, page=6]{figures/fig_more.pdf}
      \end{subfigure}%
   \hfill
    \begin{subfigure}{.25\textwidth}
   \centering %
   \includegraphics[width=\columnwidth, page=1]{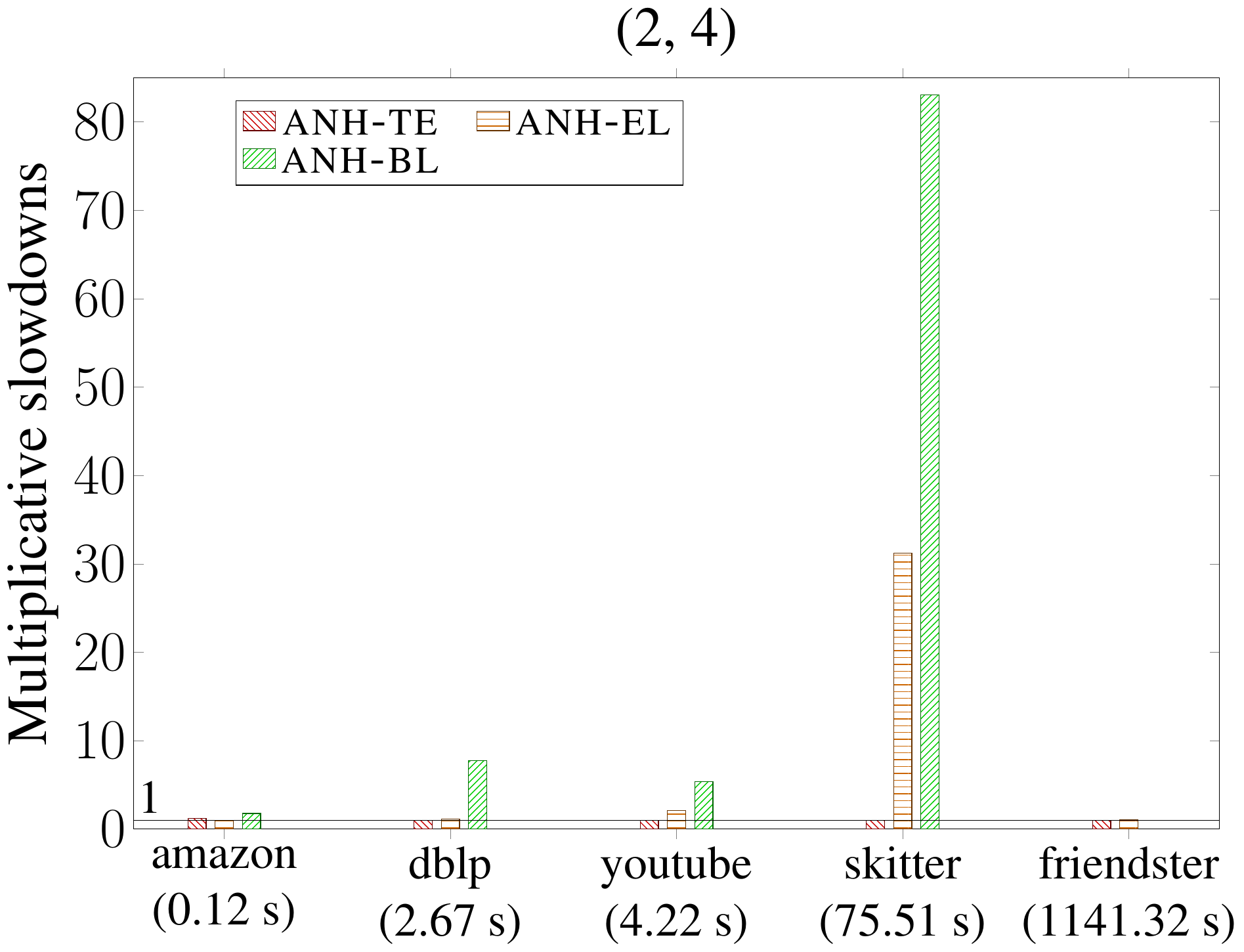}
      \end{subfigure}%

          \begin{subfigure}{.16\textwidth}
   \centering %
   \includegraphics[width=\columnwidth, page=2]{figures/fig_more.pdf}
      \end{subfigure}%
       \hfill
   \begin{subfigure}{.33\textwidth}
   \centering %
   \includegraphics[width=\columnwidth, page=7]{figures/fig_more.pdf}
   \end{subfigure}%
   \hfill
    \begin{subfigure}{.16\textwidth}
   \centering 
   \includegraphics[width=\columnwidth, page=3]{figures/fig_more.pdf}
      \end{subfigure}%
      \hfill
    \begin{subfigure}{.22\textwidth}
   \centering 
   \includegraphics[width=\columnwidth, page=4]{figures/fig_more.pdf}
      \end{subfigure}
   \vspace{-5pt}
   \caption{Multiplicative slowdowns of our parallel nucleus decomposition hierarchy implementations \ourte{}, \ourel{}, and \ourbl{}, over the fastest of the three for each graph, for $r < s \leq 5$. We have omitted bars where our implementations run out of memory or time out after 4 hours, and we have omitted graphs where only one of \ourte{}, \ourel{}, and \ourbl{} completes. Below each graph in parentheses is the fastest running time among the three implementations. We have also included a line marking a multiplicative slowdown of 1. 
   }
    \label{fig:type-1}
\end{figure*}

\subsection{Comparison of \ourte{}, \ourel{}, and \ourbl{}} \label{sec:eval-compare}

Figure~\ref{fig:type-1} compares our exact $(r, s)$ nucleus decomposition hierarchy algorithms, \ourte{}, \ourel{}, and \ourbl{}, against each other, for various $r$ and $s$. We show $r < s \leq 5$ here, but we ran all of our algorithms for $r < s \leq 7$. 
Also, the running times listed in these figures do not include the time needed to compute the low out-degree orientation  or to compute the initial $s$-clique-degrees of each $r$-clique, which are the same across all of our algorithms (note that we do include these times when comparing to other work in Section \ref{sec:exp-exact-perf}). However, our running times do include the time required to compute the $(r, s)$-clique-core numbers of each $r$-clique, which notably for \ourte{} is given by \ournd{} from~\cite{shi2021nucleus}.

Overall, we find that \ourel{} is faster if the difference between $s$ and $r$ is small (generally, if $s - r \leq 2$), and \ourte{} is faster in all other cases. 
The exception is for $k$-core (or (1, 2)-nucleus) decomposition, where \ourte{} is 2.38--21.95x faster than \ourel{}. 
This is because the $k$-core decomposition requires far lower overhead to compute the core numbers per vertex compared to higher $r$ and $s$, since we need only maintain the degree of each vertex. 
The benefit of \ourel{} is due to the improved locality in iterating over and processing $s$-cliques once, rather than recomputing the $s$-cliques twice. This is not a benefit for $k$-core, because iterating over edges is a much simpler and cache-friendly pattern. 
Also, \ourbl{} is significantly slower than both \ourte{} and \ourel{}, and runs out of memory for many values of $r$ and $s$, since it has a much larger memory footprint from storing a union-find structure per core number. 
Overall, \ourel{} is up to 2.37x faster than \ourte{}, and \ourte{} is up to 41.55x faster than \ourel{}, where we see the largest speedups in \ourte{} over \ourel{} when $s$ is much larger than $r$. \ourbl{} is up to 14.55x slower than \ourel{}, and up to 11.96x slower than \ourte{}.

The cases in which \ourel{} outperforms \ourte{} and vice versa, and the reason for the slowness of \ourbl{}, is due to the number of \ourlink{} and \algname{unite} operations. 
Indeed, for the dblp and youtube graphs, particularly for larger $r$ and when the difference between $r$ and $s$ is small, the number of times in which \ourte{} calls \ourlink{} and \algname{unite} is 1.08--13.67x the number of times in which \ourel{} calls \ourlink{} and \algname{unite}.
For smaller $r$ and when the difference between $r$ and $s$ is large, \ourel{} calls \ourlink{} and \algname{unite} between 1.02--18.94x more than \ourte{}.
Looking at fixed $r$ and increasing $s$, we observe that \ourel{} performs many more cascading calls to itself as $s$ increases, since $\binom{s}{r}$ increases and \ourel{} more likely needs to connect two $r$-cliques across multiple levels of the hierarchy tree.
\ourbl{} is much slower than both \ourte{} and \ourel{}, performing up to 39.75x the number of \ourlink{} and \algname{unite} calls. \ourbl{} repeatedly performs \algname{unite} operations equal to the core number of each $r$-clique, which can be redundant, since if two $r$-cliques are connected in a higher core, they are necessarily connected in the lower core.

In terms of memory usage, considering the memory overhead of building and constructing the hierarchy (not including the space required to store the graph or the $(r, s)$-clique-core numbers of each $r$-clique), on dblp and youtube, \ourbl{} uses 1.53--10.03x the amount of overhead that \ourel{} uses, and \ourte{} uses 1.08--1.11x the amount of overhead that \ourel{} uses. We observe that \ourel{} is the most memory efficient overall, since it only maintains two arrays proportional to the number of $r$-cliques ($uf$ and $L$). 
\ourte{} incurs almost the same memory overhead, with a minor additional cost attributed to maintaining $r$-cliques sorted by their core numbers (which we discuss in 
Section~\ref{sec:ndh-practice}),
and \ourbl{} incurs much more overhead to maintain union-find data structures proportional to the number of $r$-cliques per core number.

To provide some context into how long the hierarchy construction takes compared to the coreness computation, 
we measured the time for just computing coreness values and compared it to the total time for each of our algorithms. 
The coreness computation time represents 46.53\%, 35.26\%, and 36.08\% of the total time on average in
\ourel{}, \ourbl{}, and \ourte{}, respectively.
Note that this ignores the locality effects of interleaving the computations in \ourel{} and \ourbl{}. 
Nevertheless, we see that on average, \ourel{} spends the least amount of time on hierarchy construction, indicating the effectiveness of our optimizations. 

\begin{figure*}[t]
\vspace{-3pt}
   \begin{minipage}{0.5\textwidth}
   \centering
   \includegraphics[width=\textwidth, page=3]{figures/fig.pdf}
   \vspace{-20pt}
   \caption{
   Multiplicative slowdowns of parallel \ourndh{} for each $(r,s)$ combination over the fastest running time for parallel \ourndh{} across all $r < s \leq 7$ for each graph, considering the fastest of \ourte{}, \ourel{}, and \ourbl{}. The fastest  time is labeled in parentheses below each graph. We have omitted bars where \ourndh{} runs out of memory or times out after 4 hours.
   }
    \label{fig:all-times}
   \end{minipage}\hfill
   \begin{minipage}{0.47\textwidth}
     \centering
    \centering
   \begin{subfigure}{.5\textwidth}
   \centering
   \includegraphics[width=\columnwidth, page=1]{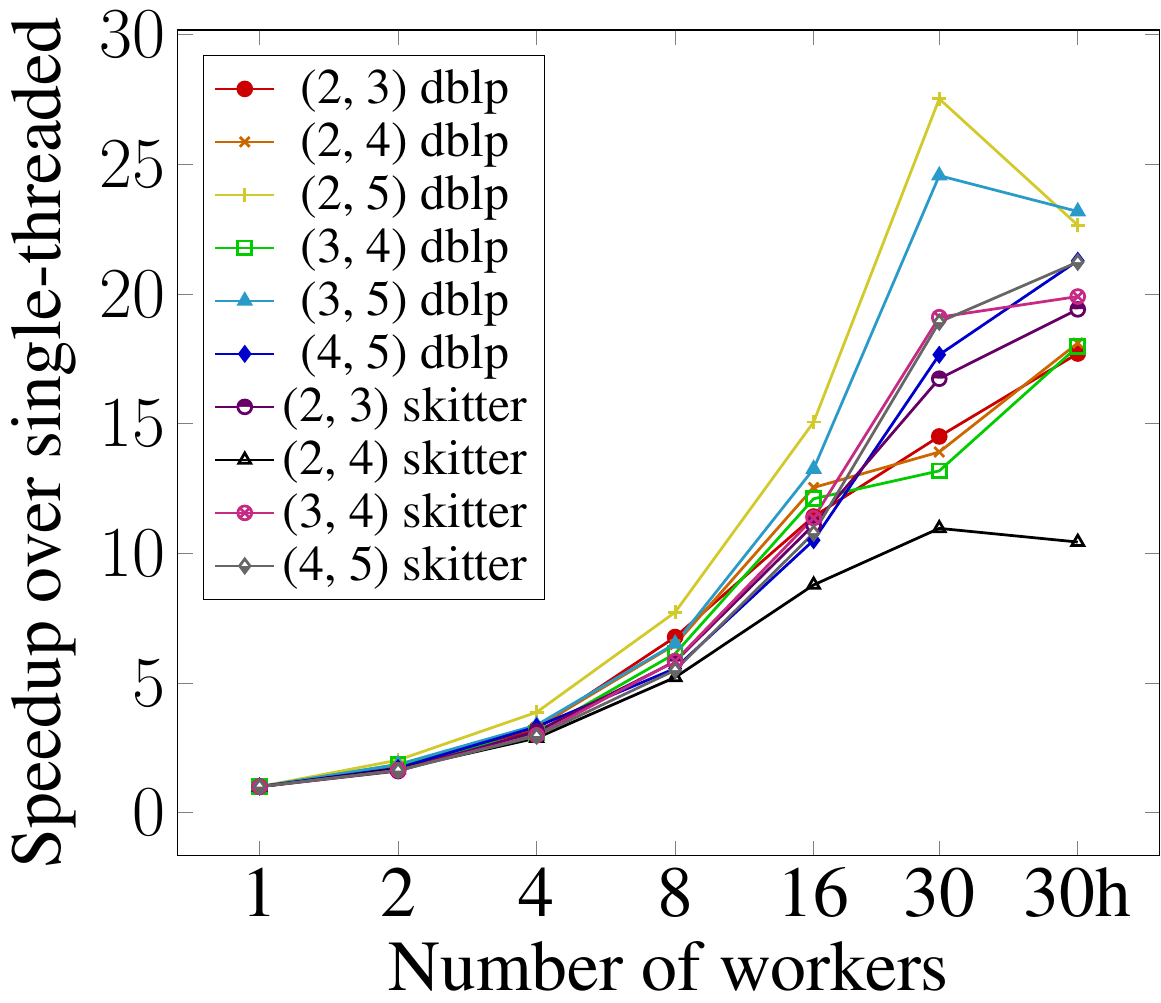}
   \end{subfigure}%
   \begin{subfigure}{.5\textwidth}
   \centering
   \includegraphics[width=\columnwidth, page=2]{figures/fig.pdf}
      \end{subfigure}
   \caption{Speedup of \ourte{} on the left and \ourel{} on the right over their respective single-threaded running times, on dblp and skitter for various $r$ and $s$. ``30h'' denotes 30-cores with two-way hyper-threading.
   }
    \label{fig:scalability}
   \end{minipage}
\end{figure*}

\subsection{Performance of Exact Hierarchy} \label{sec:exp-exact-perf}

Figure~\ref{fig:all-times} shows the best running times for all graphs, over $r < s \leq 7$, considering all of our exact $(r, s)$ nucleus decomposition hierarchy algorithms, \ourte{}, \ourel{}, and \ourbl{}, excluding
the time needed to compute our low out-degree orientation and the initial $s$-clique-degrees of each $r$-clique. 
In general, larger $(r,s)$ values correspond to longer running times.
However, some of the times for larger values of $(r,s)$ are faster than for smaller values of $(r,s)$ (especially on amazon) because the maximum coreness values for the larger values of $(r,s)$ are small and the algorithms finish quickly.

Figure~\ref{fig:scalability} shows the scalability of \ourte{} and \ourel{} over different numbers of threads on dblp and skitter, and we see good scalability overall. Across all of our graphs and for $r < s \leq 7$, we observe up to 24.75x self-relative speedups (and a median of 15.57x) for \ourte{} and up to 30.96x self-relative speedups (and a median of 14.13x) for \ourel{}. 
Generally, we observe greater self-relative speedups for larger $r$ and $s$ and for larger graphs.

\begin{figure*}[t]
\vspace{-5pt}
    \centering
   \begin{subfigure}{.33\textwidth}
   \centering
   \includegraphics[width=\columnwidth, page=4]{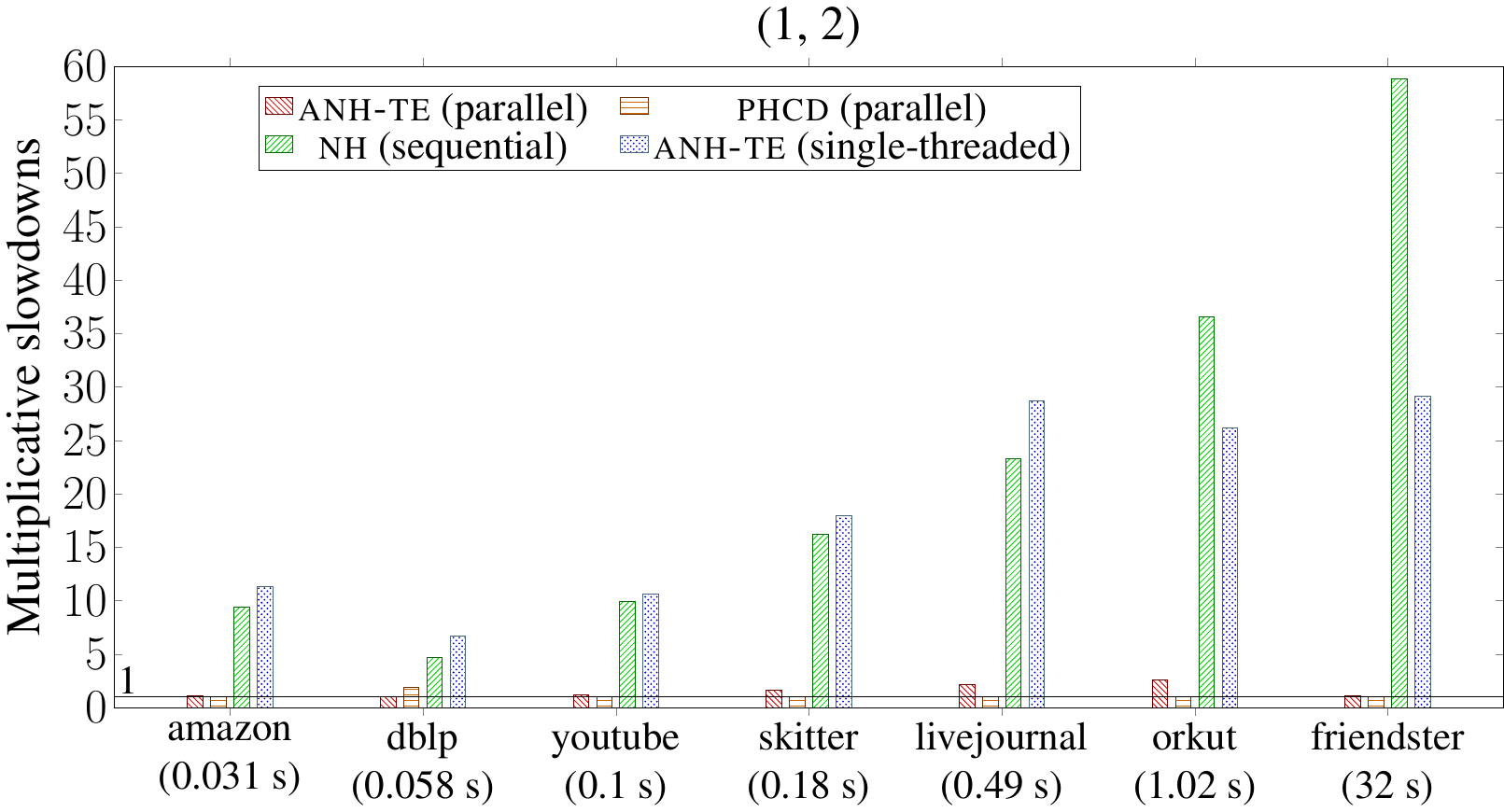}
   \end{subfigure}%
   \hfill
   \begin{subfigure}{.33\textwidth}
   \centering
   \includegraphics[width=\columnwidth, page=5]{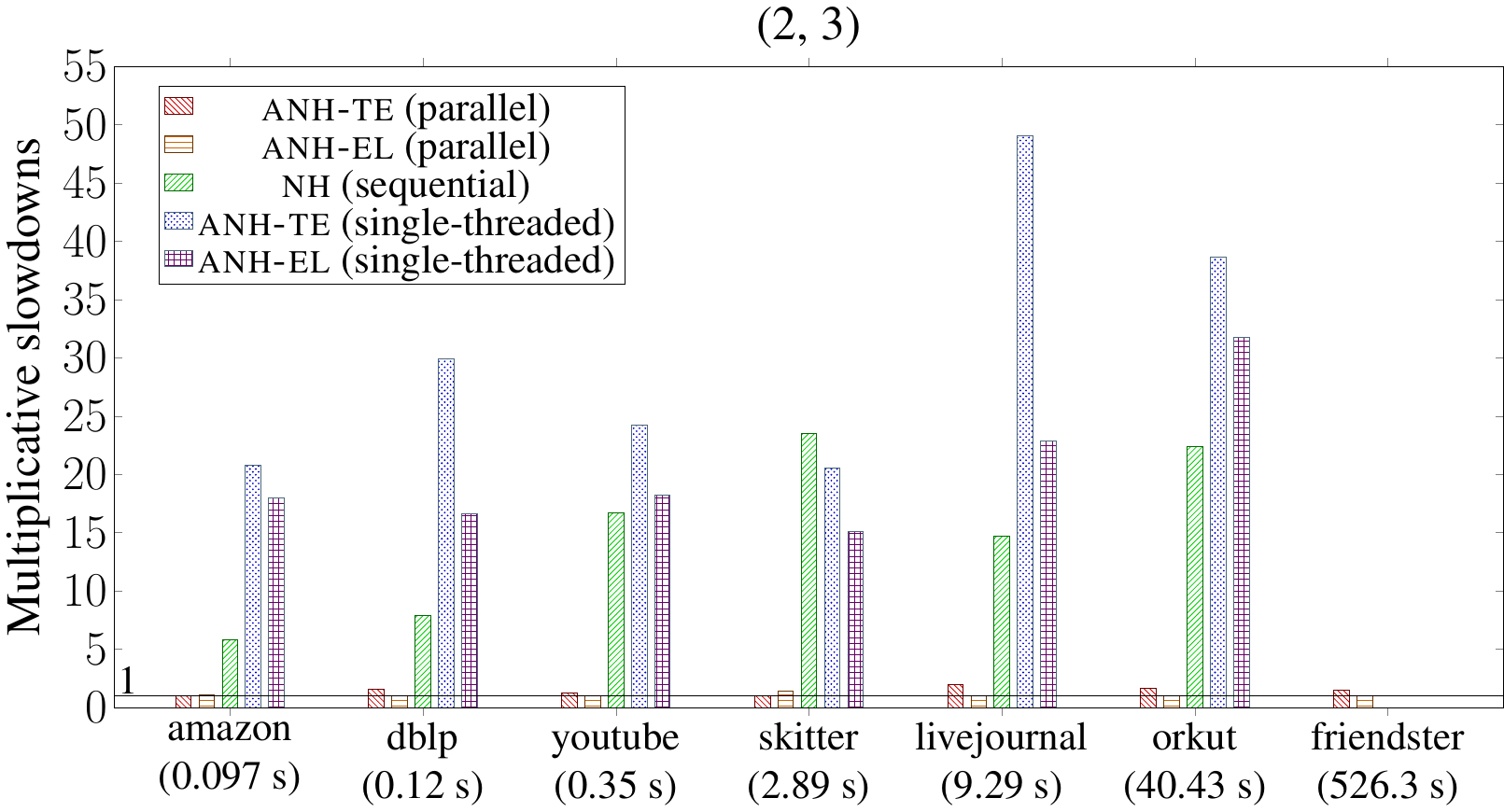}
      \end{subfigure}%
   \hfill
    \begin{subfigure}{.28\textwidth}
   \centering
   \includegraphics[width=\columnwidth, page=6]{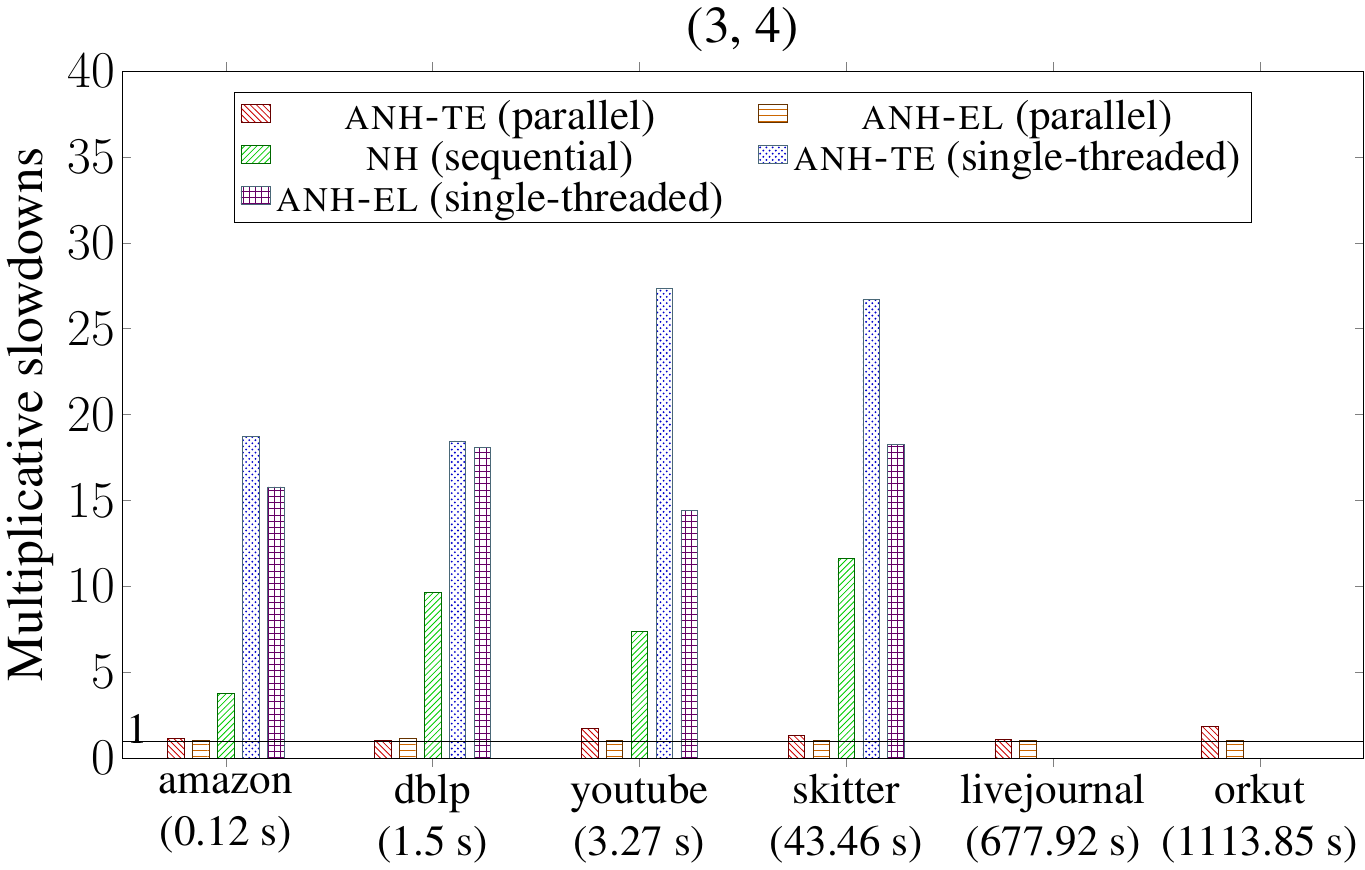}
      \end{subfigure}
      \vspace{-5pt}
   \caption{ Multiplicative slowdowns, comparing \ourte{}, \ourel{}, the parallel \algname{phcd}~\cite{Chu2022seqkcore}, and the sequential \algname{nh}~\cite{SaPi16}, for $(1, 2)$, $(2, 3)$, and $(3, 4)$ nucleus decomposition. 
   We give the multiplicative slowdown over the fastest implementation for each graph and each $(r, s)$, where the fastest running time is labeled in parentheses below each graph.
   We include end-to-end running times in this comparison, excluding only the time required to load the graph. We have omitted bars where the implementation runs out of memory or times out after 4 hours.
   We have also included a line marking a multiplicative slowdown of 1. 
   }
    \label{fig:comp-prior}
\end{figure*}

\myparagraph{Comparison to Other Implementations}
Figure~\ref{fig:comp-prior} shows the comparison of our parallel $(1, 2)$, $(2, 3)$, and $(3, 4)$ nucleus decomposition hierarchy implementations to other implementations. Note that here, we include in our implementations the time needed to compute the low out-degree orientation and to compute the initial $s$-clique-degrees of each $r$-clique. We do not include the time required to load the graph in both our and other implementations.

For $(1, 2)$ nucleus ($k$-core) decomposition, we compare to the parallel \algname{phcd}~\cite{Chu2022seqkcore} and the sequential \algname{nh}~\cite{SaPi16}. 
Our fastest implementation for $k$-core is \ourte{}, and we see that \ourte{} is up to 2.57x slower than \algname{phcd} overall, but 1.87x faster than \algname{phcd} on dblp. 
We note that \algname{phcd} is optimized for the $k$-core decomposition, rather than general $(r, s)$ nucleus decomposition, whereas \ourte{} generalizes for larger $r$ and $s$.
Like \ourte{}, \algname{phcd} constructs the hierarchy tree from the bottom-up after computing the core numbers of each vertex, but unlike \ourte{}, \algname{phcd} leverages the information from computing the core numbers to optimize for the $k$-core hierarchy, by reordering vertices based on their core numbers. This optimization allows them to more efficiently divide the work of constructing the hierarchy across different threads, and allows them to reduce the work in practice when iterating over the neighbors of a vertex $v$ with larger core numbers than that of $v$.
Compared to the sequential \algname{nh}, \ourte{} is 4.67--58.84x faster, particularly on larger graphs.

For $(1,2)$, $(2, 3)$, and $(3, 4)$ nucleus decomposition, we compare our parallel \ourte{} and \ourel{} to the sequential \algname{nh}~\cite{SaPi16}. Considering the fastest of \ourte{} and \ourel{} for each graph, our implementations are 3.76--23.54x faster, demonstrating that we achieve good speedups from our use of parallelization.
Sequentially, our fastest algorithm is between 2.02x faster and 4.2x slower than \algname{nh}.
We get slowdowns due to the additional overheads of parallel subroutines and the fact that our code has a general structure for supporting arbitrary $(r,s)$ values, whereas the code of~\cite{SaPi16} is specialized for the specific values of $(1,2)$, $(2,3)$, and $(3,4)$. These slowdowns are somewhat counteracted by our faster clique enumeration algorithm, leading to speedups in some cases.

\begin{figure}[t]
    \centering
    \hfill
   \begin{subfigure}{.35\textwidth}
   \centering
   \includegraphics[width=\columnwidth]{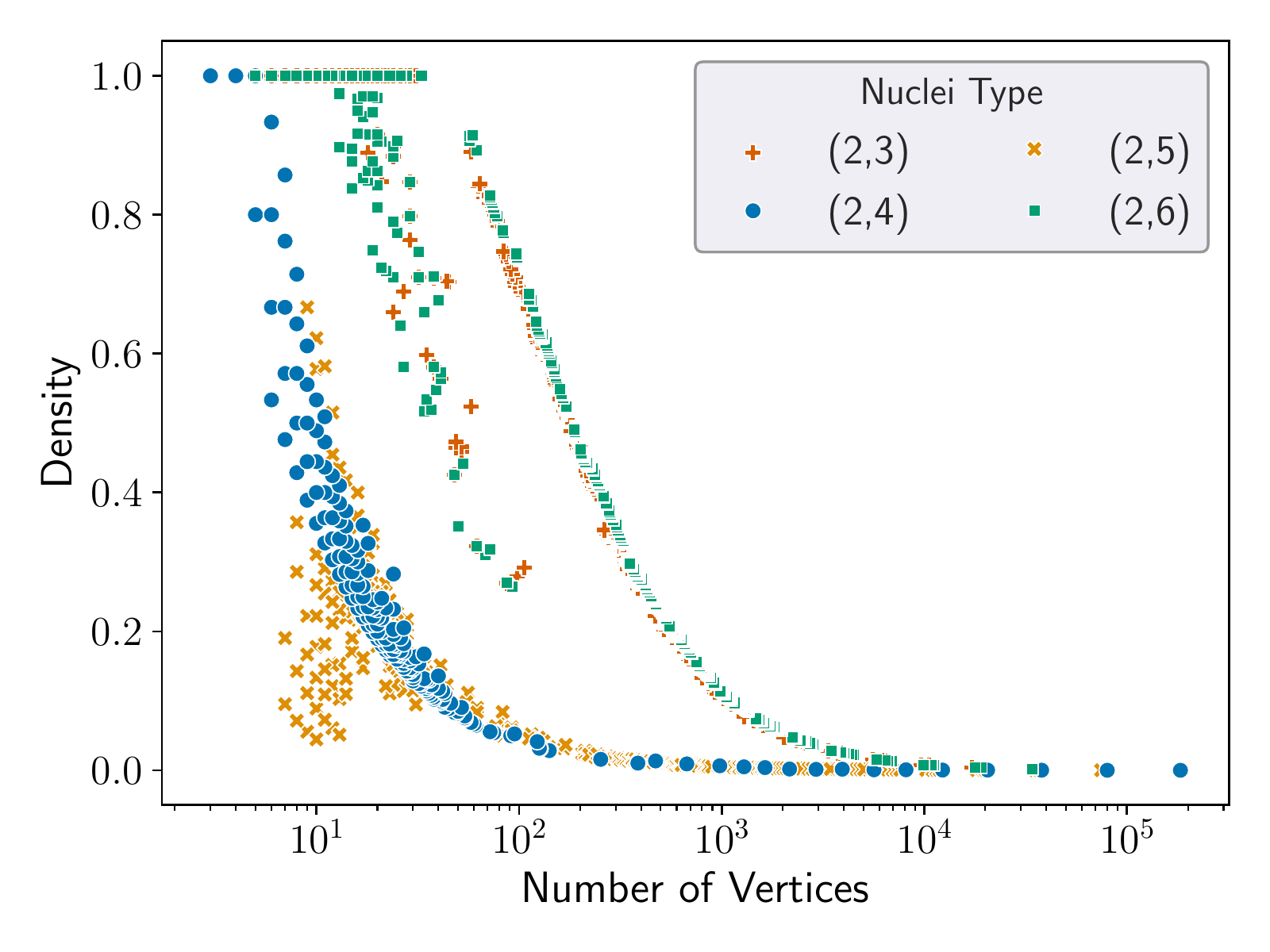}
   \end{subfigure}%
   \hfill
   \begin{subfigure}{.35\textwidth}
   \centering
   \includegraphics[width=\columnwidth]{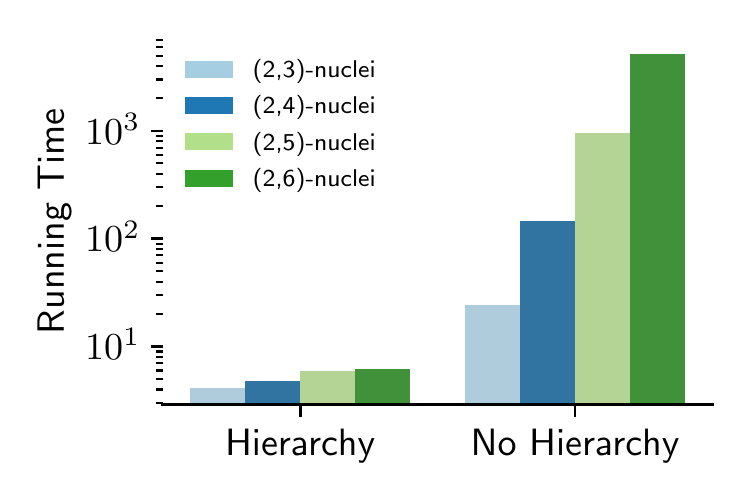}
      \end{subfigure}%
      \hfill
      \vspace{-10pt}
      \caption{(Left) Number of vertices vs.\ edge density of subgraphs from various $(2,s)$-nuclei on the youtube graph. (Right) Running times for finding all $(2,s)$-nuclei with and without the hierarchy.}
\label{fig:density}
\end{figure}

\myparagraph{Usefulness of the Hierarchy}
In Figure~\ref{fig:density},
we show the edge densities of various $(2,s)$-nuclei in the youtube graph. The \emph{edge density} of a subgraph $S$
is defined to be the number of edges in $S$ divided by $|S|\choose 2$ (the total possible number of edges). 
Edge density is the metric used to evaluate the quality of subgraphs from nucleus decompositions~\cite{Sariyuce2017}, and
finding dense subgraphs subject to a size constraint has important applications.

For a fixed $r$ and $s$,
having the nucleus hierarchy enables us to efficiently find all of the
$c$-$(r,s)$ nuclei for any value of $c$ by 
simply removing all internal nodes corresponding to $c'$-$(r,s)$ nuclei for $c'<c$,
and having each resulting subtree be a $c$-$(r,s)$ nucleus that contains all of its leaves.
On the other hand, given only the core numbers and no hierarchy, to find the $c$-$(r,s)$ nuclei for a given value of $c$, one would need to run connected components on all $r$-cliques with core number $\geq c$, with connectivity defined by $s$-clique adjacency, which is much more expensive than simply cutting the hierarchy. 
Indeed, Figure~\ref{fig:density} shows the time needed to generate all $(2,s)$ nuclei for various values of $s$ with and without the hierarchy for the youtube graph. We see that using the hierarchy is 5.84--834.09x faster than not using the hierarchy.

\subsection{Performance of Approximate Hierarchy}
We considered $\delta = 0.1$, $0.5$, and $1$ for
our experiments for approximate $(r, s)$ nucleus decomposition ($\delta$ is the approximation parameter in Algorithm~\ref{alg:ndh-approx}).
We first compare \ourapproxnd{} to \ournd{}~\cite{shi2021nucleus}
and see a speedup of up to 16.16x for $\delta = 0.1$, up to 8.35x for $\delta = 0.5$, and up to 10.88x for $\delta = 1$. 

Besides the computation of the $(r, s)$-clique-core numbers,
\ourapproxte{}, \ourapproxel{}, and \ourapproxbl{} are identical to \ourte{}, \ourel{}, and \ourbl{}, respectively. In other words, their hierarchy construction procedure is the same. We observe up to a 3.3x speedup considering the fastest of our approximate algorithms for each graph and $r < s \leq 7$, over the fastest of our exact algorithms. Notably, for $\delta = 0.1$, we are able to compute the $(2, 5)$ nucleus decomposition hierarchy on friendster in 8783.2 seconds, where our exact implementations timeout at 4 hours. 
The improvements in running time using our approximate algorithms are lower than when comparing \ourapproxnd{} to \ournd{}~\cite{shi2021nucleus} because
even in our approximate algorithms, $s$-cliques must be exactly counted per $r$-clique, and much of the time is spent 
doing this.

In terms of accuracy, the average error in the $(r, s)$-clique-core number per $r$-clique is relatively low for all of our $\delta$ values. For $\delta = 0.1$, across all of our graphs and for $r < s \leq 7$, our coreness estimates per $r$-clique are have a multiplicative error of 1--2.92x on average (with a median of 1.33x) compared to the exact coreness numbers.
For $\delta = 0.5$, the coreness estimates range from having a multiplicative error of 1--2.92x on average as well, with a median of 1.34x, and for $\delta = 1$, the coreness estimates range from having a multiplicative error of 1--3.05x on average, with a median of 1.35x.
The multiplicative errors of the maximum $(r, s)$-clique-core number, across all graphs and for $r < s \leq 7$, are also reasonably low, with a median of 1.6x for $\delta = 0.1$, a median of 2x for $\delta = 0.5$, and a median of 2x for $\delta = 1$. The maximum multiplicative error for a given $r$-clique is 6.73x for $\delta = 0.1$, 6.98x for $\delta = 0.5$, and 7.32x for $\delta = 1$, but these arise for large $s$, notably when $r = 5$ and $s = 7$ for all $\delta$, and these errors are still much lower than the theoretical guarantee of $(\binom{s}{r} + \delta) \cdot (1 + \delta)$ given by Theorem~\ref{thm:approx}.
Note that the theoretical bound is a worst case bound.
In practice, $(r,s)$-cliques whose core numbers 
are close to the end of the range of each bucket will have a much better approximation. 
Furthermore, graphs with a low peeling complexity will likely have more $(r,s)$-cliques peeled before moving to the next round (Line~\ref{alg-approxnd:bound-rounds} of Algorithm~\ref{alg:ndh-approx}), which will make the approximation better.

\section{Conclusion}\label{sec:conclusion}
We have presented new parallel exact and approximate algorithms for nucleus hierarchy construction with strong theoretical guarantees. We have developed optimized implementations of our algorithms, which interleave the coreness number computation with the hierarchy construction. Our experiments showed that our implementations outperform state-of-the-art implementations while achieving good parallel scalability.

\myparagraph{Acknowledgements}
This research was supported by NSF Graduate Research Fellowship
\#1122374, 
DOE Early Career Award \#DE-SC0018947,
NSF CAREER Award \#CCF-1845763, Google Faculty Research Award, Google Research Scholar Award, 
cloud computing credits from Google-MIT, FinTech@CSAIL Initiative, DARPA
SDH Award \#HR0011-18-3-0007, and Applications Driving Architectures
(ADA) Research Center, a JUMP Center co-sponsored by the Semiconductor Research Corporation (SRC) and DARPA.

\bibliographystyle{ACM-Reference-Format}
\bibliography{references}

\end{document}